\documentclass[12pt]{article}%
\usepackage{fullpage}
\usepackage{amsmath, amsfonts, amssymb, amsthm}
\usepackage{mathtools}
\usepackage{mathrsfs}
\usepackage{latexsym}
\usepackage{mathtools}
\usepackage{apptools}
\usepackage{array}
\usepackage{orcidlink}

\usepackage{color}
\usepackage{graphicx}  
\usepackage{hyperref}
\usepackage{subfig}

\usepackage{wasysym}
\usepackage{float}
\usepackage{enumitem}
\usepackage{scalerel}
\usepackage[linesnumbered,ruled,vlined]{algorithm2e}
\usepackage[utf8]{inputenc}
\usepackage{listings}

\definecolor{darkpink}{rgb}{0.91, 0.33, 0.5}
\definecolor{caribbeangreen}{rgb}{0.0, 0.8, 0.6}
\definecolor{cobalt}{rgb}{0.0, 0.28, 0.67}
\definecolor{amethyst}{rgb}{0.6, 0.4, 0.8}

\newtheorem{theorem}{Theorem}
\newtheorem{claim}{Claim}
\newtheorem{lemma}{Lemma}
\newtheorem{corollary}{Corollary}
\newtheorem{definition}{Definition}

\newtheorem{proposition}{Proposition}

\newcommand{\C}{\mathcal{C}}

\DeclareFontShape{OT1}{cmr}{m}{scit}{<-> ssub * cmr/m/sc}{}

\title{Cluster Editing on Cographs and Related Classes}

\author{{Manuel Lafond\orcidlink{0000-0002-5305-7372}\footnote{Department of Computer Science, Université de Sherbrooke, \textit{Email: manuel.lafond@usherbrooke.ca}}   } \and {Alitzel L\'opez S\'anchez\orcidlink{0000-0002-3545-039X}\footnote{Department of Computer Science, Université de Sherbrooke, \textit{Email: alitzel.lopez.sanchez@usherbrooke.ca}}} \and  {Weidong Luo\orcidlink{0009-0003-5300-606X}\footnote{Department of Computer Science, Université de Sherbrooke, \textit{Email: weidong.luo@yahoo.com}}}}

\begin{document}

\maketitle
\begin{abstract}
In the \textsc{Cluster Editing} problem, sometimes known as (unweighted) \textsc{Correlation Clustering}, we must insert and delete a minimum number of edges to achieve a graph in which every connected component is a clique.  Owing to its applications in computational biology, social network analysis, machine learning, and others, this problem has been widely studied for decades and is still undergoing active research. There exist several parameterized algorithms for general graphs, but little is known about the complexity of the problem on specific classes of graphs.
    
Among the few important results in this direction, if only deletions are allowed, the problem can be solved in polynomial time on cographs, which are the $P_4$-free graphs. However, the complexity of the broader editing problem on cographs is still open. We show that even on a very restricted subclass of cographs, the problem is NP-hard, W[1]-hard when parameterized by the number $p$ of desired clusters, and that time $n^{o(p/\log p)}$ is forbidden under the ETH.  This shows that the editing variant is substantially harder than the deletion-only case, and that hardness holds for the many superclasses of cographs (including graphs of clique-width at most $2$, perfect graphs, circle graphs, permutation graphs).  On the other hand, we provide an almost tight upper bound of time $n^{O(p)}$, which is a consequence of a more general $n^{O(cw \cdot p)}$ time algorithm, where $cw$ is the clique-width.  Given that forbidding $P_4$s maintains NP-hardness, we look at $\{P_4, C_4\}$-free graphs, also known as trivially perfect graphs, and provide a cubic-time algorithm for this class. 
\end{abstract}

\section{Introduction}

Clustering objects into groups of similarity is a ubiquitous task in computer science, with applications in computational biology~\cite{DBLP:journals/dam/ShamirST04,lafond2018accurate}, 
social network analysis~\cite{abu2021greedy, arrighi2023cluster, veldt2018correlation}, machine learning~\cite{dhanachandra2015image,ezugwu2022comprehensive}, and many others~\cite{becker2005survey}.  There are many interpretations of what a ``good'' clustering is, 
with one of the most simple, elegant, and useful being the \textsc{Cluster Editing} formulation --- sometimes also known as (unweighted) \textsc{Correlation Clustering}~\cite{bansal2004correlation}.  In this graph-theoretical view, pairs of objects that are believed to be similar
are linked by an edge, and non-edges correspond to dissimilar objects.  
If groups are perfectly separable, this graph should be a \emph{cluster graph}, that is, a graph in which each connected component is a clique. 
However, due to noise and errors, this is almost never observed in practice.  To remove such errors, 
\textsc{Cluster Editing} asks for a minimum number of edges to correct 
to obtain a cluster graph, where ``correcting'' means adding a non-existing edge or deleting an edge.


Owing to its importance, this APX-hard \cite{DBLP:conf/focs/CharikarGW03}, of course also NP-hard~\cite{bansal2004correlation,  DBLP:journals/acta/KrivanekM86}, problem has been widely studied in the parameterized complexity community. Let $k$ be the number of required edge modifications. After a series of works, the problem now can be solved in time $O^*(1.62^k)$  \cite{DBLP:journals/jda/Bocker12, DBLP:journals/tcs/BockerBBT09, DBLP:journals/ipl/BockerD11, DBLP:conf/ciac/GrammGHN03, DBLP:journals/algorithmica/GrammGHN04} and admits a $2k$ kernel~\cite{DBLP:journals/algorithmica/CaoC12, DBLP:journals/jcss/ChenM12,  DBLP:journals/tcs/Guo09}. In addition, if we require that the solution contains exactly $p$ clusters, then the problem is NP-hard for every $p \geq 2$ \cite{DBLP:journals/dam/ShamirST04}, but admits a PTAS \cite{DBLP:journals/toc/GiotisG06}, a $(p+2)k + p$ kernel \cite{DBLP:journals/tcs/Guo09}, and can be solved in $2^{O(\sqrt{pk})}n^{O(1)}$ time. 
This is shown to be tight assuming the ETH, under which $2^{o(\sqrt{pk})}n^{O(1)}$ is not possible~\cite{DBLP:journals/jcss/FominKPPV14}.

Another angle, which we study in this paper, is to focus on specific classes of graphs. 
For example, restricting the input to bounded-degree graphs does not help, as \textsc{Cluster Editing} is NP-hard even on planar unit disk graphs with
maximum degree $4$~\cite{komusiewicz2012cluster, ochs2023cluster}.  In~\cite{berger2016ptas}, the authors circumvent the APX-hardness of the problem by proposing a PTAS on planar graphs. A polynomial-time algorithm is provided for the problem on unit interval graphs~\cite{mannaa2010cluster}, a subclass of unit disk. 
The \textsc{Cluster Deletion} problem, in which only edge deletions are allowed, has received much more attention on restricted classes. It is polynomial-time solvable on graphs of maximum degree $3$ and NP-hard for larger degrees~\cite{komusiewicz2012cluster}.  Various results were also obtained on interval and split graphs~\cite{konstantinidis2021cluster}, other subclasses of chordal graphs~\cite{bonomo2015complexity}, and unit disk graphs~\cite{ochs2023cluster}.   In~\cite{italiano2023structural}, graphs with bounded structural parameters are studied, with the weighted variant being paraNP-hard in twin cover, but FPT in the unweighted case.

If we forbid specific induced subgraphs, the reduction in~\cite{komusiewicz2012cluster} implies that \textsc{Cluster Deletion} is NP-hard on $C_4$-free graphs (as observed in~\cite{DBLP:journals/dm/GaoHN13}).  If instead we forbid $P_4$, i.e., induced paths on four vertices, we obtain the class of \emph{cographs}, on which the deletion problem is remarkably shown to be polynomial-time solvable in~\cite{DBLP:journals/dm/GaoHN13} using a greedy max-clique strategy. However, the complexity of \textsc{Cluster Editing} on cographs has remained open. 
In addition, to our knowledge, there are no known non-trivial polynomial-time algorithms for \textsc{Cluster Editing} on specific graph classes with a finite set of forbidden induced subgraphs. It is not hard to obtain a polynomial-time algorithm for the problem on the \emph{threshold graphs}, i.e., the $\{P_4, C_4, 2K_2\}$-free graphs, using dynamic programming on its specific co-tree.
However, it appears to be unknown whether removing any of the three induced subgraphs from this set leads to NP-hardness.



In this paper, we focus on open complexity questions for the \textsc{Cluster Editing} problem on cographs and related classes. 
It is worth mentioning that the cograph restriction is more than a mere complexity classification endeavor --- it can be useful to determine how well an equivalence relation (i.e., a cluster graph) can approximate a different type of relation (see for example~\cite{zahn1964approximating}).  In the case of cograph-like relations, our motivations have roots in phylogenetics.  
In this field, gene pairs are classified into \emph{orthologs} and \emph{paralogs}, 
with orthology graphs known to correspond exactly to
cographs~\cite{hellmuth2013orthology,lafond2014orthology,hellmuth2015phylogenomics}.
However, as argued in~\cite{sanchez2021colorful}, most orthology prediction software use clustering libraries and infer a cluster graph of orthologs.  The question that arises is 
then ``how much information is lost by predicting a cluster graph, knowing that the true graph is a cograph''?  This requires finding a cluster graph that is the closest to a given cograph $G$, leading to \textsc{Cluster Editing} on cographs.  
Furthermore, researchers argue that social communities should sometimes be modeled as cographs \cite{jia2015defining} or trivially perfect graphs \cite{DBLP:journals/socnet/NastosG13}, as opposed to cluster graphs as is done in most community detection approaches. This leads to \textsc{Cluster Editing} on cographs or trivially perfect graphs. 
Additionally, an algorithm for the NP-hard problem \textsc{Trivial Perfect Graph Editing}, which is the first one capable of scaling effectively to large real-world graphs, is provided by \cite{DBLP:conf/esa/BrandesHSW15}.

\vspace{3mm}

\noindent
\textbf{Our contributions.}  
We first settle the complexity of \textsc{Cluster Editing} on cographs by showing that it is not only NP-hard, but also W[1]-hard when a parameter $p$ is specified, which represents the 
number of desired clusters.  We use the \textsc{Unary Bin Packing} hardness results from~\cite{DBLP:journals/jcss/JansenKMS13}, which 
also implies an Exponential Time Hypothesis (ETH) lower bound that forbids time $n^{o(p/ \log p)}$ under this parameter.  
In fact, our hardness holds for very restricted classes of cographs, namely graphs obtained by taking two cluster graphs, and adding all possible edges between them (this also correspond to cographs that admit a cotree of height $3$).
Moreover, because cographs have \emph{clique-width} ($cw$) at most 2, this also means that the problem is para-NP-hard in clique-width, and that 
a complexity of the form $n^{g(cw)}$ is unlikely, for any function $g$ (the same actually holds for the \emph{modular-width} parameter and generalizations, see~\cite{gajarsky2013parameterized,lafond_et_al:LIPIcs.MFCS.2023.61}).  In fact, the ETH forbids time $f(p)n^{g(cw)\cdot o(p/ \log p)}$ for any functions $f$ and $g$, which contrasts with the aforementioned subexponential bounds in $pk$~\cite{DBLP:journals/jcss/FominKPPV14}.

The hardness also extends to all superclasses of cographs, such as circle graphs, perfect graphs, and permutation graphs.
On the other hand we show that time $n^{O(p)}$ can be achieved on any cograph, which is almost tight.  This contrasts with the general \textsc{Cluster Editing} problem which is NP-hard when $p = 2$.
In fact, this complexity follows from a more general algorithm for arbitrary graphs that runs in time $n^{O(cw \cdot p)}$, which shows that \textsc{Cluster Editing} is XP in parameter $cw + p$.  Note that our hardness results imply that XP membership in either parameter individually is unlikely, and so under standard assumptions both $cw$ and $p$ must contribute in the exponent of $n$. 

Finally, we aim to find the largest subclass of cographs on which \textsc{Cluster Editing} is polynomial-time solvable.  
The literature mentioned above implies that such a class lies somewhere between $P_4$-free and $\{P_4, C_4, 2K_2\}$-free graphs.  
We improve the latter by showing that \textsc{Cluster Editing} can be solved in time $O(n^3)$ on
$\{P_4, C_4\}$-free graphs, also known as trivially perfect graphs (TPG).  This result is achieved by a characterization of optimal clusterings on TPGs, 
which says that as we build a clustering going up in the cotree, only the largest cluster is allowed to become larger as we proceed.

Our results are summarized in the following, where $n$ is the number of vertices of the graphs, and $p$-\textsc{Cluster Editing} is the variant in which the edge modifications must result in $p$ connected components that are cliques.  We treat $p$ as a parameter specified in the input.
\begin{theorem}
The following results hold:
\begin{itemize}
    \item 
    \textsc{Cluster Editing} is NP-complete on cographs, and solvable in time $O(n^3)$ on trivially perfect graphs.
   \item 
    $p$-\textsc{Cluster Editing} admits an $n^{O(cw\cdot p)}$ time algorithm if a $cw$-expression is given, but admits no $f(p)n^{g(cw)\cdot o(p/ \log p)}$ time algorithm for any functions $f$ and $g$ unless ETH fails.

    \item 
    $p$-\textsc{Cluster Editing} on cographs is NP-complete, and W[1]-hard parameterized by $p$.

    \item 
    $p$-\textsc{Cluster Editing} on cographs admits an $n^{O(p)}$ time algorithm, but admits no $f(p)n^{o(p/ \log p)}$ time algorithm for any function $f$ unless ETH fails.

\end{itemize}
\end{theorem}

\section{Preliminaries}

We use the notation $[n] = \{1,\ldots,n\}$. 
For two sets $A$ and $B$, $A \bigtriangleup B$ is the symmetric difference between $A$ and $B$. 
For a graph $G$, $V(G)$ and $E(G)$ are the vertex and edge sets of $G$, respectively, and $G[S]$ is the subgraph induced by $S \subseteq V(G)$.
The complement of $G$ is denoted $\overline{G}$.
Given two graphs $G$ and $H$, the  \emph{disjoint union} $G\cup H$ is the graph with vertex set $V(G) \cup V(H)$ and edge set $E(G) \cup E(H)$.
The \emph{join} $G \vee H$ of two graphs 
$G$ and $H$ is the graph obtained from $G \cup H$ by adding every possible edge $uv$ with $u \in V(G)$ and $v \in V(H)$.


It will be useful to consider two equivalent views on the \textsc{Cluster Editing} problem, in terms of the edge operations to perform to achieve a cluster graph, and in terms of the resulting partition into clusters.
A graph is a \textit{cluster graph} if it is a disjoint union of complete graphs. Let $G = (V, E)$ be a graph and $F \subseteq  V \times V$. If $G' = (V,  E \bigtriangleup F)$ is a cluster graph, then $F$ is called a \emph{cluster editing set}. The edges of $F$ can be divided into two types:  $F \cap E(G)$ are called \emph{deleted edges}, and $F \setminus E(G)$ are called \emph{inserted edges}, where the deleted edges \emph{disconnect} some adjacent vertices in $G$ and the inserted edges \emph{connect} some non-adjacent vertices in $G$ to transform $G$ into $G'$.

Note that the clusters of $G'$ result in a partition of $V(G)$.  Conversely, given any partition $\C$ of $V(G)$, we can easily infer the editing set $F$ that yields the clusters $\C$: it consists of non-edges within the clusters, and of edges with endpoints in different clusters.  To be precise, a \emph{clustering} of $G$ is a partition $\C = \{C_1, \ldots, C_l\}$ of $V(G)$.  The \emph{cluster editing set} of $\C$ is
\[
edit(\C) := \{ uv \in E(G) : u \in C_i, v \in C_j, i \neq j\} \cup \bigcup_{i \in [l]} E(\overline{G[C_i]} ).
\]

We define $cost_G(\mathcal{C}) = |edit(\C)|$. 
An element of $C\in \mathcal{C}$ is called a \emph{cluster}, and the cardinality of $\mathcal{C}$ is called \emph{cluster number}. An \emph{optimal cluster editing set} for $G$ is a cluster editing set for $G$ of minimum size. An \emph{optimal clustering} is a partition $\C$ of $V(G)$ of minimum cost.

A formal definition of \textsc{Cluster Editing} problem is as follows.
\vspace{2mm}

\noindent \textsc{Cluster Editing} \\
\noindent \emph{Input}: A graph $G$ and an integer $k$.\\
\noindent \emph{Question}: Is there a clustering of $G$ with cost at most $k$?

\vspace{2mm}

A clustering with exactly $p$ clusters is called a $p$-\emph{clustering}. 
In \textsc{$p$-Cluster Editing}, the problem is the same, but we must find a $p$-clustering of cost at most $k$. 




We will sometimes use the fact that twins, which are pairs of vertices that have the same closed neighborhood, can be assumed to be in the same cluster.   More generally, a \emph{critical clique} of a graph $G$ is a clique $K$ such that all $v \in K$ have the same neighbor vertices in $V(G)\setminus K$, and $K$ is maximal under this property.

\begin{proposition}[\cite{DBLP:journals/jcss/ChenM12, DBLP:journals/tcs/Guo09}]
\label{critical-clique-prop}
Let $K$ be a critical clique of $G$. For any optimal clustering $\mathcal{C}$ of $G$, there is a $C \in \mathcal{C}$ such that $K \subseteq C$.
\end{proposition}

\vspace{2mm}

\noindent
\textbf{Cographs and cotrees.}
A \emph{cograph} is a graph that can be constructed using the three following rules: (1) a single vertex is a cograph; (2) the disjoint union $G \cup H$ of cographs $G, H$ is a cograph; (3) the join $G \vee H$ of cographs $G, H$ is a cograph.  
Cographs are exactly the $P_4$-free graphs, i.e., that do not contain a path on four vertices as an induced subgraph.

Cographs are also known for their tree representation.  For a graph $G$, a \emph{cotree} for $G$ is a rooted tree $T$ whose set of leaves, denoted $L(T)$, satisfies $L(T) = V(G)$.  Moreover, every internal node $v \in V(T) \setminus L(T)$ is one of two types, either a $0$-node or a $1$-node, such that $uv \in E(G)$ if and only if the lowest common ancestor of $u$ and $v$ in $T$ is a $1$-node\footnote{To emphasize the distinction between general graphs and trees, we will refer to vertices of a tree as \emph{nodes}.}.  It is well-known that $G$ is a cograph if and only if there exists a cotree $T$ for $G$.  This can be seen from the intuition that leaves represent applications of Rule (1) above, $0$-node represents disjoint unions, and $1$-node represents joins.

A \emph{trivially perfect graph} (TPG), among several characterizations, is a cograph $G$ that has no induced cycle on four vertices, i.e., a $\{P_4, C_4\}$-free graph.  TPGs are also the chordal cographs.  For our purposes, a TPG is a cograph that admits a cotree $T$ in which every $1$-node has at most one child that is not a leaf (see~\cite[Lemma 4.1]{heggernes2006linear}).

\vspace{2mm}

\noindent
\textbf{Clique-width and NLC-width.}
Our clique-width ($cw$) algorithm does not use the notion of $cw$ directly, but instead the analogous measure of \emph{NLC-width}~\cite{heggernes2006linear}.  We only provide the definition of the latter.
Let $G = (V, E, lab)$ be a graph in which every vertex has one of $k$ node labels ($k$-NL), where $lab$ is a function from $V$ to $[k]$.  A single-vertex graph labeled $t$ is denoted by $\bullet_t$.
\begin{definition}
A $k$-node labeled controlled ($k$-NLC) graph is a $k$-NL graph defined recursively as follows. 
\begin{enumerate}
    \item $\bullet_t$ is a $k$-NLC graph for every $t\in [k]$.
    \item  Let $G_1 = (V_1, E_1, lab_1), G_2 = (V_2, E_2, lab_2)$ be two $k$-NLC graphs, relation $S \subseteq [k]^2$, and $E_{add} = \{uv : u\in V_1 \wedge v\in V_2 \wedge (lab_1(u), lab_2(v)) \in S\}$. 
    
    The $k$-NL graph $G = (V, E, lab)$ denoted $G_1 \times _S G_2$ is a $k$-NLC graph defined as: $V = V_1 \cup V_2$, $E = E_1 \cup E_2 \cup E_{add}$, and $lab(u) = lab_1(u)$, $lab(v) = lab_2(v)$ for $u\in V_1$, $v \in V_2$.
    
    \item Let $G = (V, E, lab)$ be a $k$-NLC graph and $R$ be a function from $[k]$ to $[k]$. Then $\circ_R(G) = (V, E, lab')$ is a $k$-NLC graph, where $lab'(v) = R(lab(v))$ for all $v\in V$.
\end{enumerate}
\end{definition}
Intuitively, operation 2 denoted by $\times_S$ takes the disjoint union of the graphs $G_1, G_2$, then adds all possible edges between labeled vertices of $G_1$ and labeled vertices of $G_2$ as controlled by the pairs of $S$. Operation 3 denoted by $\circ_R$ relabels the vertices of a graph controlled by $R$. $NLC_k$ denotes all $k$-NLC graphs. The \emph{NLC-width} of a labeled graph $G$ is the smallest $k$ such that $G$ is a $k$-NLC graph. Furthermore, for a labeled graph $G = (V, E, lab)$, the \emph{NLC-width} of a graph $G' = (V, E)$, obtained from $G$ with all labels removed, equals the \emph{NLC-width} of $G$.

A well-parenthesized expression $X$ built with operations 1, 2, 3 is called a $k$-\emph{expression}. The graph constructed by $X$ is denoted by $G_X$. We associate the $k$-\emph{expression tree} $T$ of $G_X$ with $X$ in a natural way, that is, leaves of $T$ correspond to all vertices of $G_X$ and the internal nodes of $T$ correspond to the operations 2, 3 of $X$. For each node $u$ of $T$, the sub-tree rooted at $u$ corresponds to a \emph{sub} $k$-\emph{expression} of $X$ denoted by $X_u$, and the graph $G_{X_u}$ constructed by $X_u$ is also denoted by $G^u$. Moreover, we say $G^u$ is the \emph{related graph} of $u$ in $T$.  

Let us briefly mention that a graph has clique-width at most $k$ if it can be built using what we will call a $cw(k)$-expression, which has four operations instead: creating a single labeled vertex $\bullet_t$; taking the disjoint union; adding all edges between vertices of a specified label pair $(i, i')$; relabeling all vertices with label $i$ to another label $j$.
We do not detail further, since the following allows us to use NLC-width instead of clique-width.

\begin{proposition}[\cite{DBLP:journals/tcs/GurskiWY16, johansson1998clique}]
\label{cw-nlcw-relation}
For every graph $G$, the clique-width of $G$ is at least the NLC-width of $G$, and at most two times the NLC-width of $G$. 
Moreover, any given $cw(k)$-expression can be transformed into an equivalent $k$-expression in polynomial time (i.e., yielding the same graph).
\end{proposition}

We will assume that our $k$-expressions are derived from a given $cw(k)$-expression. 
Reference~\cite{johansson1998clique} is often cited for this transformation, but seems to have vanished from nature.  The transformation can be done using the normal form of a $cw(k)$-expression described in~\cite{espelage2003deciding}.

\section{Cluster editing on cographs}

We first prove our hardness results for \textsc{Cluster Editing} and $p$-\textsc{Cluster Editing} on cographs, using a reduction from \textsc{Unary Bin Packing}. An instance of \textsc{Unary Bin Packing} consists of a unary-encoded integer multiset $A= [a_1,\ldots,a_n]$, which represent the sizes of $n$ items, and integers $b,k$. We must decide whether the items can be packed into at most $k$ bins that each have capacity $b$.
Thus, we must partition $A$ into $k$ multisets $A_1, \ldots, A_k$, some possibly empty, such that $\sum_{a \in A_i} a \leq b$ for each $i \in [k]$.

For our purpose, we introduce a variant of this problem called \textsc{Unary Perfect Bin Packing}.  The problem is identical, except that the partition of $A$ into $A_1, \ldots, A_k$ must satisfy $\sum_{a \in A_i}a = b$ for each $i \in [k]$.  That is, we must fill all $k$ bins to their maximum capacity $b$.  Note that for this to be possible, $\sum_{i=1}^n a_i = kb$ must hold.
Note that these packing problems can be solved in polynomial time for any fixed $k$ \cite{DBLP:journals/jcss/JansenKMS13}.  We assume henceforth that $k, n \geq 10$, $\max_{i=1}^{n}{a_i} \leq b$, and $\sum_{i=1}^{n} a_i \leq kb$, otherwise, they can be decided in polynomial time. 
It is known that \textsc{Unary Bin Packing} is NP-complete \cite{DBLP:books/fm/GareyJ79}, W[1]-hard parameterized by the number of bins $k$ \cite{DBLP:journals/jcss/JansenKMS13}, and does not admit an $f(k)|I|^{o(k/ \log k)}$ time algorithm for any function $f$ unless Exponential Time Hypothesis (ETH) fails \cite{DBLP:journals/jcss/JansenKMS13}, where $|I|$ is the size of the instance (in unary). The results easily extend to the perfect version.

\begin{proposition}
\label{upbp-is-hard}
\textsc{Unary Perfect Bin Packing} is NP-complete, W[1]-hard parameterized by $k$, and has no $f(k)|I|^{o(k/ \log k)}$ time algorithm for any function $f$ unless ETH fails. 
\end{proposition}

\begin{proof}
Clearly, \textsc{Unary Perfect Bin Packing} is in NP. Let $(A,b,k)$ be an instance of \textsc{Unary Bin Packing}, where $A = [a_1,\ldots,a_n]$. We assume that $k$ is even, as otherwise we increase $k$ by $1$ and add to $A$ an item of value $b$.
If $\sum_{i=1}^{n} a_i \leq 0.1kb$, we argue that we have a YES instance: we can pack the largest $k/2$ of the $n$ items into $k/2$ bins, and pack the other $n - k/2$ items in the remaining bins, as follows.  
Assume w.l.o.g. $b\geq a_1 \geq \cdots \geq a_{k/2} \geq a_{k/2+1} \geq \cdots \geq a_{n}$. We have $k/2 \cdot a_{k/2}\leq \sum_{i=1}^{k/2} a_i \leq 0.1kb$, so $a_{k/2} \leq 0.2b$. This means that $a_{k/2 +1}, \ldots, a_n \leq 0.2 b$ and we can always select a batch of items with these sizes such that the total size of a batch is in the interval $[0.8b, b]$ until there are no items left (the last batch may have a size less than $0.8b$). Thus, we can pack the items with sizes $a_{k/2 +1}, \ldots, a_n$ using at most $\frac{0.1kb}{0.8b} + 1 = 0.125k + 1 < k/2$ bins.

So assume henceforth that $\sum_{i=1}^{n} a_i > 0.1kb$. 
Construct an instance $(B,b,k)$ for \textsc{Unary Perfect Bin Packing}, where $B = [a_1,\ldots,a_n, 1,\ldots,1]$ consists of all integers of $A$ and $kb-\sum_{i=1}^{n} a_i$ 1s. Since $0.1kb < \sum_{i=1}^{n} a_i \leq kb$, the new instance size is bounded by a linear function of the original instance size. Moreover, the new instance can be obtained in polynomial time. It is easy to verify that $(A,b,k)$ is a YES instance of \textsc{Unary Bin Packing} if and only if $(B,b,k)$ is a YES instance of \textsc{Unary Perfect Bin Packing}. 
\end{proof}


\begin{lemma}
\label{ce-is-w-hard}
Given a unary-encoded integer multiset $A= [a_1,\ldots,a_n]$ for the sizes of $n$ items, and integers $b,k$ satisfying $kb = \sum_{i=1}^{n} a_i$, there is a polynomial-time algorithm which outputs a cograph $G$ and an integer $t$ such that, 
\begin{enumerate}
    \item for any optimal clustering $\mathcal{C}$ of $G$, $|\mathcal{C}| = k$ and the cost of $\mathcal{C}$ is at least $t$;
    \item the $n$ items can be perfectly packed by $k$ bins with capacity $b$ if and only if there is a clustering of $G$ with cost at most $t$.
\end{enumerate}
Moreover, $G$ is obtained by taking a join of two cluster graphs.
\end{lemma}

\begin{proof}
For the remainder, let us denote 
$a := \sum_{i=1}^{n} a_i$ and $s := \sum_{i=1}^{n} a_i^2$. Construct an instance $(G,t)$ for \textsc{Cluster Editing} as follows. First, add two cluster graphs $I = B_1\cup \ldots \cup B_k$ and $J = A_1\cup \ldots \cup A_n$ into $G$, where $B_i$ is a complete graph with $h := (nka)^{10}$ vertices for every $i\in [k]$, and $A_j$ is a complete graph with $a_j$ vertices for every $j\in [n]$. Then, connect each $v\in V(I)$ to all vertices of $V(J)$. See Figure~\ref{fig:cograph-hard}.  One can easily verify that $G$ is a cograph obtained from the join of two cluster graphs. In addition, define $t := (k-1)ah + \frac{1}{2} (kb^2 - s)$. Clearly, $t$ is an integer since $s \equiv a \equiv kb \equiv kb^2 $ (mod 2).
Let $\mathcal{I} = \{V(B_1), \dots, V(B_k)\}$ and $\mathcal{J} = \{V(A_1), \dots, V(A_n)\}$. Clearly, every element in $\mathcal{I} \cup \mathcal{J}$ is a critical clique of $G$. 

\begin{figure}[ht]
    \centering
    \includegraphics[width=0.5\linewidth]{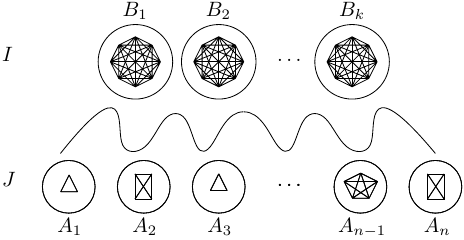}
    \caption{An illustration of the construction.  In the subgraph $I$, each $B_i$ is a ``large enough'' complete graph, and in the subgraph $J$, each $A_j$ is a complete graph with $a_j$ vertices.  The wiggly line indicates that all possible edges between $I$ and $J$ are present (there are no edges between two $B_i$'s, and no edge between two $A_j$'s).}
    \label{fig:cograph-hard}
\end{figure}

\begin{claim}
\label{opt-CES-is-t}
Let $\mathcal{C}$ be an optimal clustering of $G$, and  $F$ be the cluster editing set for this solution. Then, $|\mathcal{C}| = k$, and each element of $\mathcal{C}$ is a superset of exactly one element from $\mathcal{I}$. Moreover, $|F|\geq t$, and $|F| = t$ if and only if all elements of $\mathcal{C}$ have the same cardinality.
\end{claim}

\begin{proof}
Let $\mathcal{C} = \{P_1, \ldots, P_l\}$. Since each element in $\mathcal{I} \cup \mathcal{J}$ is a critical clique of $G$, each such element is a subset of some cluster in $\mathcal{C}$ by Proposition \ref{critical-clique-prop}. Note that each cluster of $\mathcal{C}$ could contain $0, 1$, or more cliques of $\mathcal{I}$.  We split the possibilities into three cases and provide bounds on $|F|$ for each case.
First, if there exists a cluster of $\mathcal{C}$ that includes at least two critical cliques from $\mathcal{I}$, then $F$ contains $h^2$ inserted edges to connect two critical cliques from $\mathcal{I}$, and thus $|F| \geq h^2$. 
Assume instead that every cluster of $\mathcal{C}$ includes at most one critical clique from $\mathcal{I}$. Then, we have $l \geq k$. 
Suppose  $l \geq k + 1$. Then, there are $l - k$ clusters in $\mathcal{C}$ that do not contain any critical cliques from $\mathcal{I}$.  Let $\mathcal{C}' \subseteq \mathcal{C}$ consist of these $l - k$ clusters and $U = \bigcup_{P \in \mathcal{C}'} P$. Then, for every $v\in U$, $kh$ deleted edges are required in $F$ to disconnect $v$ from all vertices of $V(I)$, and for every $u\in V(J)\setminus U$, $(k-1)h$ deleted edges are required in $F$ to disconnect $u$ from all vertices of $V(I)\setminus P$, where $P$ is the clique of $\mathcal{I}$ contained in the same cluster as $u$. Therefore, 
\begin{align*}
|F|  & \geq kh|U| + (k-1)h |V(J)\setminus U| \\
 & = h|U| + (k -1)h|U|  + (k-1)h |V(J)\setminus U|\\
 & = h|U| + (k-1)ha\\
 & \geq (k-1)ah + h.
\end{align*}

Now assume that $l = k$. Then, every element of $\mathcal{C}$ includes exactly one critical clique from $\mathcal{I}$. Consider each $i\in [k]$ and assume w.l.o.g. that $V(B_i) \subseteq P_i$. Let $W_i = P_i \setminus V(B_i)$ and let $\{\mathcal{J}_1, \ldots, \mathcal{J}_k\}$ be a partition  
of $\mathcal{J}$ such that, for each $i \in [k]$, the union of all elements of $\mathcal{J}_i$ is $W_i$ (such a partition exists because each element of $\mathcal{J}$ is entirely contained in some $W_i$). 
Firstly, $(k-1)h$ deleted edges are required in $F$ to disconnect each $v\in W_i$ from all vertices of $V(I) \setminus V(B_i)$.  Secondly, for each $i \in [k]$, $\frac{1}{2}\sum_{S, S' \in \mathcal{J}_i} |S||S'|$ inserted edges are required in $F$ to connect all vertices of $W_i$. 
One can easily check that for each $i \in [k]$, $F$ accounts for every edge with an endpoint in $P_i$ and an endpoint outside, and accounts for all non-edges within $P_i$.
Therefore, all the possible edges of $F$ are counted.  Thus,
\begin{align*}
|F| &= \sum_{i=1}^{k} \left(|W_i| (k-1)h + \frac{1}{2}\sum_{ S, S' \in \mathcal{J}_i} |S||S'| \right) \\
&=   |V(J)| (k-1) h + \frac{1}{2} \sum_{i=1}^{k} \left( \left(\sum_{ S \in \mathcal{J}_i} |S| \right)^2 - \sum_{ S \in \mathcal{J}_i} |S|^2\right)  \\
&=  (k-1)ah + \frac{1}{2} \sum_{i=1}^{k} |W_i|^2 - \frac{1}{2} \sum_{ S \in \mathcal{J}} |S|^2  \\
&=  (k-1)ah + \frac{1}{2} \sum_{i=1}^{k} |W_i|^2 - \frac{s}{2}.
\end{align*}
We can now compare $|F|$ from the lower bounds obtained in the previous two cases, as follows.
\begin{align*}
|F| & = (k-1)ah + \frac{1}{2} \sum_{i=1}^{k} |W_i|^2 - \frac{s}{2}\\
   & <  (k-1)ah +  \sum_{i=1}^{k} |W_i|^2 + \sum_{1\leq i, j\leq k} |W_i| |W_j|\\
   & =  (k-1)ah +  \left(\sum_{i=1}^{k} |W_i|\right)^2\\
   & =  (k-1)ah +  a^2\\
   & <  (k-1)ah + h < h^2.
\end{align*}
The last line implies that having $|\mathcal{C}| = k$, with each element of $\mathcal{C}$ a superset of exactly one element from $\mathcal{I}$, always achieves a lower cost than the other possibilities.

Next, consider the lower bound of $|F| \geq t$ and the conditions on equality.  Using the same starting point, 
\begin{align}
|F| &=   (k-1)ah + \frac{1}{2} \sum_{i=1}^{k} |W_i|^2 - \frac{s}{2}  \nonumber \\
&=   (k-1)ah + \frac{1}{2k} \left(\sum_{i=1}^{k} |W_i|^2\right)\left(\sum_{i=1}^{k} 1^2\right) - \frac{s}{2}  \nonumber \\
&\geq    (k-1)ah + \frac{1}{2k} (\sum_{i=1}^{k} |W_i|)^2 - \frac{s}{2} \label{optimal_equal_t}\\
&=   (k-1)ah + \frac{1}{2k} a^2 - \frac{s}{2} = t.   \nonumber
\end{align}
In (1), we used the Cauchy-Schwarz inequality, and the two sides are equal if and only if $|W_1| = \cdots = |W_k|$.
In addition, $|W_1| = \cdots = |W_k|$ if and only if $|P_1| = \cdots = |P_k|$ (recall that $W_i = P_i \setminus V(B_i)$ and $|V(B_i)| = h$ for each $i \in [k]$).  This proves every statement of the claim.
\end{proof}

Armed with the above claim, we immediately deduce the first statement of the theorem.  We now show the second part, i.e., the $n$ items of $A$ can fit perfectly into $k$ bins of capacity $b$ if and only if we can make $G$ a cluster graph with at most $t$ edge modifications.

\medskip

\noindent 
$(\Leftarrow)$: Assume there is a clustering $\mathcal{C}$ of $G$ with cost at most $t$. 
By Claim \ref{opt-CES-is-t}, the size of any optimal cluster editing set for $G$ is at least $t$. 
Thus, the cost of $\mathcal{C}$ must be equal to $t$, and $\mathcal{C}$ must be optimal.  Claim~\ref{opt-CES-is-t} then implies that $\mathcal{C} = \{P_1, \ldots, P_k\}$, 
such that $|P_1| = \cdots = |P_k|$ and $V(B_i) \subseteq P_i$ for each $i \in [k]$ (w.l.o.g.).
Moreover, each critical clique of $\mathcal{J}$ is a subset of some element of $\mathcal{C}$ by Proposition \ref{critical-clique-prop}. So there is a partition $\mathcal{J}_1, \ldots, \mathcal{J}_k$ of $\mathcal{J}$ such that $P_i \setminus V(B_i) = \bigcup_{S\in \mathcal{J}_i} S$ for every $i \in [k]$. This means that $\mathcal{J}$ can be partitioned into $k$ groups such that the number of vertices in each group equals $b$. 

\medskip

\noindent 
$(\Rightarrow)$: Assume the $n$ items can be perfectly packed into $k$ bins with capacity $b$.  We construct a cluster editing set $F$ of size at most $t$. 
Consider each $i \in [k]$.
Let $S_i \subseteq A$ consist of the sizes of the items packed in the $i$-th bin, and define a corresponding cluster graph $B'_i = \bigcup_{a_j \in S_i} A_j$.  
We put in $F$ a set of $(k-1)h$ deleted edges, by disconnecting each $v\in V(B'_i)$ from every vertex in $V(I)\setminus V(B_i)$ in $G$.  We then add to $F$ the set of $\frac{1}{2} \sum_{V(A_j), V(A_r) \subseteq V(B'_i)} a_j a_r$ of inserted edges into $F$ to make $B_i'$ a complete graph, where $j, r \in [n]$. 
Applying the modifications in $F$ results in a clustering $\mathcal{C} = \{P_1, \ldots, P_k\}$ such that $P_i = V(B_i) \cup V(B'_i)$ for each $i\in [k]$.  Moreover, the cardinality of the cluster editing set is
\begin{align*}
|F| &=  \sum_{i=1}^{k} \left((k - 1)h |V(B'_i)| + \frac{1}{2} \sum_{V(A_j), V(A_r) \subseteq V(B'_i)} a_ja_r \right) \\
&=   (k - 1)h \sum_{i=1}^{n} |V(A_i)|  + \frac{1}{2}
\sum_{i=1}^{k} \left( \left(\sum_{V(A_j) \subseteq V(B'_i)} a_j \right)^2 - \sum_{V(A_j) \subseteq V(B'_i)} a_j^2 \right)\\
&=   (k - 1)h |V(J)|  + \frac{1}{2}
\sum_{i=1}^{k} \left( |V(B'_i)|^2 - \sum_{V(A_j) \subseteq V(B'_i)} a_j^2 \right) \\
&=   (k - 1)ah  + \frac{1}{2}
\sum_{i=1}^{k} b^2  -  \frac{1}{2} \sum_{V(A_j) \subseteq V(J)} a_j^2 \\
&=   (k - 1)ah  + \frac{1}{2}
k b^2  -  \frac{1}{2} s = t.
\end{align*}
Thus, $G$ has a clustering with cost at most $t$. 
\end{proof}

\begin{theorem}
\label{hardness-theorem-ce-cograph}
The following hardness results hold:
\begin{itemize}
    \item 
    \textsc{Cluster Editing} on cographs and $p$-\textsc{Cluster Editing} on cographs are NP-complete.

    \item 
    $p$-\textsc{Cluster Editing} on cographs is W[1]-hard parameterized by $p$.

    \item 
    $p$-\textsc{Cluster Editing} on cographs admits no $f(p)n^{o(p/ \log p)}$ time algorithm for any function $f$ unless ETH fails, where $n$ is the number of vertices of the input graph.
\end{itemize}
\end{theorem}

\begin{proof}
\textsc{Cluster Editing} and $p$-\textsc{Cluster Editing} are clearly in NP. 
Lemma \ref{ce-is-w-hard} directly implies a Karp reduction from \textsc{Unary Perfect Bin Packing} to \textsc{Cluster Editing} on cographs. Thus, \textsc{Cluster Editing} on cographs is NP-hard by Proposition \ref{upbp-is-hard}.
By item 1 of Lemma \ref{ce-is-w-hard}, for a partition $\mathcal{C}$ of $V(G)$, $\mathcal{C}$ is a clustering of $G$ with cost at most $t$ if and only if $\mathcal{C}$ is a $k$-clustering of $G$ with cost at most $t$, so Lemma \ref{ce-is-w-hard} also implies 
a parameter-preserving Karp reduction from \textsc{Unary Perfect Bin Packing} in parameter $k$ to $p$-\textsc{Cluster Editing} on cographs in parameter $p$ such that $p = k$.
Therefore, combining the results in Proposition \ref{upbp-is-hard}, we have $p$-\textsc{Cluster Editing} on cographs is NP-hard, W[1]-hard parameterized by $p$, and has no $f(p)n^{o(p/ \log p)}$ time algorithm for any function $f$ unless ETH fails.
\end{proof}

Since cographs have a constant clique-width~\cite{courcelle2000upper}, we have the following.

\begin{corollary}
$p$-\textsc{Cluster Editing} admits no $f(p)n^{g(cw)\cdot o(p/ \log p)}$ time algorithm for any functions $f$ and $g$ unless ETH fails.
\end{corollary}

\subsection*{An $n^{O(p \cdot cw)}$ time algorithm}

We propose a dynamic programming algorithm over the $k$-expression tree $T$ of $G$ as described in the preliminaries, where $k$ is at most twice the clique-width of $G$.  The idea is that, for each node $u$ of $T$ and every way of distributing the vertices of $V(G^u)$ among $p$ clusters and $k$ labels, we compute the minimum-cost clustering conditioned on such a distribution.  The latter is described as a matrix of vertex counts per cluster per label, and the cost can be computed by combining the appropriate cost for matrices of the children of $u$.

The entries of all matrices discussed here are natural numbers between $0$ and $n$. $M_{i,:}$ and $M_{:,j}$ represent the $i$-th row and $j$-the column of matrix $M$, respectively. We write $m_{i,j}$ for the entry of $M$ in row $i$ and column $j$. 
The sum of all entries in $M_{i,:}$ and $M_{:,j}$ are denoted by $sum\{M_{i,:}\}$ and $sum\{M_{:,j}\}$, respectively. We use $\{S_i\}_{i=1}^n$ to denote a sequence $(S_1, \ldots, S_n)$.

Let $M$ be a $k \times p$ matrix and $G$ be a $k$-NL graph. An $M$-\emph{sequencing} of $G$ is a sequence $\{C_i\}_{i=1}^p$ of subsets of $V(G)$ such that 
\begin{enumerate} 
    \item the non-empty subsets of this sequence form a partition $\mathcal{C}$ of $V(G)$; 
    \item the number of vertices in $C_j$ labeled $i$ is equal to $m_{i,j}$ for every $i \in [k], j \in [p]$.   
\end{enumerate}
The partition $\mathcal{C}$ obtained from $\{C_i\}_{i=1}^p$ is called an $M$-\emph{clustering}, and the \emph{cost} of the $M$-\emph{sequencing} is the cost of the clustering $\mathcal{C}$ of $G$, which is $cost_G(\mathcal{C})$.  

Clearly, an $M$-sequencing of $G$ exists if and only if the sum of all entries of $M$ equals the number of vertices of $G$ and the sum of all entries of the $i$-th row of $M$ equals the number of vertices in $G$ labeled $i$ for every $i\in [k]$. The cost of $M$-sequencing is defined as $\infty$ if it does not exist. In addition, we say $M$ is a \emph{well-defined matrix} for $G$ if an $M$-sequencing of $G$ exists.  An \emph{optimal $M$-sequencing} of $G$ is an $M$-sequencing of $G$ of minimum cost.


\begin{theorem}
$p$-\textsc{Cluster Editing} has an $O(n^{2p \cdot cw + 4})$ time algorithm if a $k$-expression is given, where $k = cw$ and  $n$ is the number of vertices of the input graph.
\end{theorem}
\begin{proof}
Let $T$ be a $k$-expression tree of $k$-NLC graph $G = (V, E, lab)$. For every $u\in V(T)$, let $G^u =(V^u, E^u, lab^u)$ be the related graph of $u$, and let $V^u_i$ denote the vertices of $V^u$ labeled $i$ for each $i\in [k]$. Let $M^u = (m^z_{i,j})$ be a well-defined matrix of $G^u$. 
Assume $u$ is a node of $T$ corresponding to operation $\times_S$, and $v, w$ are the two children of $u$.  We define $h(M^v, M^w, S)$, for later use, which equals 
\begin{align*}
&\sum_{j \in [p]}  sum\{M^v_{:,j}\} \cdot sum\{M^w_{:,j}\} +  \sum_{(i, i')\in S} sum\{M^v_{i,:}\} \cdot sum\{M^w_{i',:}\} - 2\sum_{(i, i')\in S} \sum_{j\in [p]} m^v_{i,j} m^w_{i',j}.
\end{align*}
Next, we will first provide the dynamic programming algorithm for this problem, and then prove its correctness. 


Let dynamic programming table $D(u, M^u)$ denote the cost of an optimal $M^u$-sequencing of $G^u$. For a leaf $u$ of $T$ and all well-defined $M^u$ of $u$, $D(u, M^u) : = 0$. For an internal vertex $u$ of $T$ with corresponding operation $\times_S$, and children $v, w$, we have
\begin{align*}
D(u, M^u) : =  & \min_{M^u = M^v + M^w} \{D(v, M^v) + D(w, M^w) + h(M^v, M^w, S)\}.
\end{align*}
For an internal node $u$ of $T$ with corresponding operation $\circ_R$, and child $v$, we have  
\begin{align*}
D(u, M^u) : =  & \min_{M^v} D(v, M^v),
\end{align*}
where the minimization is taken over every well-defined matrix $M^v$ for $G^v$ that satisfies $M^u_{i',:} = \sum_{(i,i')\in R} M^v_{i,:}$ for every $i'\in [k]$.

Since every entry of the $k\times p$ matrix $M^u$ is at most $n$, each $u \in V(T)$ has at most $n^{pk}$ tables. For the $\times_S$ vertices, one entry $D(u, M^u)$ 
can be computed in time $O(n^{pk + 3})$ by enumerating all the possible $O(n^{pk})$ matrices $M^v$, from which $M^w$ can be deduced from $M^u - M^v$ in time $O(pk) = O(n^2)$, and computing $h(M^v, M^w, S)$ in time $O(n^3)$ (we treat $p$ and $k$ as upper-bounded by $n$ for simplicity).  Thus the set of all entries for 
 $u$ can be computed in $O(n^{2pk +3})$ time if all tables of its children are given. For the $\circ_R$ vertices, 
 each entry can be computed in time $O(n^{pk + 3})$ similarly by enumerating the $M^v$ matrices and checking the sum conditions, for a total time of $O(n^{2pk + 3})$ as well.
Since $T$ has $O(n)$ nodes, we can compute all tables for all nodes, from leaves to root, of $T$ in 
$O(n^{2pk + 4})$ time, which is 
$O(n^{2p \cdot cw + 4})$ if we assume $cw = k$.
Let $r$ be the root of $T$. The output of our algorithm is the minimum $D(r, M^r)$ such that $M^r$ has no zero columns (note that if we require \emph{at most} $p$ clusters, we can simply tolerate zero columns).

We now prove that the recurrences are correct.  The leaf case is easy to verify, so we assume inductively that the table is filled correctly for the children of an internal node $u$.
Suppose that $u$ corresponds to operation $\times_S$ and consider the value we assign to an entry $D(u, M^u)$. Assume the two children of $u$ are $v, w$. Suppose $\{C_i^u\}_{i=1}^p$ is an optimal $M^u$-sequencing for $G^u$ with $M^u$-clustering $\mathcal{C}^u$.  
Let $\mathcal{C}^v = \{C \cap V^v : C \in \mathcal{C}^u \} \setminus \{\emptyset\}$ and $\mathcal{C}^w = \{C \cap V^w : C \in \mathcal{C}^u \} \setminus \{\emptyset\}$.
Since $G^u$ is an union of $G^v$ and $G^w$ by adding some edges between them controlled by $S$, we claim that the cost of the optimal $M^u$-sequencing $cost_{G^u}(\mathcal{C}^u)$ equals
\begin{align*}
\sum _{y\in\{v, w\}} cost_{G^y}(\mathcal{C}^y) + \sum_{C \in \mathcal{C}^u} |C \cap V^v||C \cap V^w| &+ \sum_{(i, i')\in S}|V^v_i||V^w_{i'}| \\
&- 2\sum_{(i, i')\in S} \sum_{C \in \mathcal{C}^u} |C \cap V^v_i||C \cap V^w_{i'}|, 
\end{align*} 
justified as follows. First, 
any inserted or deleted edge of $\C^u$ whose endpoints are both in $V^v$, or both in $V^w$, is counted exactly once, namely in the first summation $\sum_{y\in\{v, w\}} cost_{G^y}(\mathcal{C}^y)$.
Then, consider the inserted/deleted edges with one endpoint in $V^v$ and the other in $V^w$. Observe that a non-edge between $V^v$ and $V^w$ is counted once in the second summation if and only if it is an inserted edge.  That non-edge is not counted in the third summation, because the latter only counts edges of $G^u$, and it is not subtracted in the last term for the same reason. 
Thus the expression counts inserted edges between $V^v, V^w$ in $\C^u$ exactly once, and no other such non-edge.
Now, consider an edge $e$ between $V^v$ and $V^w$, which exists because of $S$. 
If $e$ is a deleted edge, its endpoints are in different clusters and it is counted exactly once, in the third summation.  If $e$ is not deleted, its endpoints are in the same cluster and it is counted in both the second and third summation.  
On the other hand, that edge is subtracted twice in the last term, and so overall it is not counted.  We deduce that the expression correctly represents the cost of $\C^u$ in $G^u$.

For each $y \in \{v, w\}$, we further define $M^y = (m_{i,j}^y)$ such that $m_{i,j}^y = |C_j^u \cap V^y_i|$ for all $i, j$. Clearly, we have $M^v + M^w = M^u$. Now, we claim that $\{C_i^u \cap V^y\}_{i=1}^p$ is an \emph{optimal} $M^y$-sequencing of $G^y$ with $M^y$-clustering $\mathcal{C}^y$ for every $y \in \{v, w\}$.
Roughly speaking, this can be seen by observing that merging any $M^v$ and $M^w$-clustering will result, in the cost expression given above, in the same values for the three last summations.  Since the choice of $M^y$ clusterings only affects the first summation, they should be chosen to minimize it. To see this in more detail, assume for contradiction that $\{D_i^v\}_{i=1}^p$ is an $M^v$-sequencing of $G^v$ with $M^v$-clustering $\mathcal{D}^v$ and  $\{D_i^w\}_{i=1}^p$ is an $M^w$-sequencing of $G^w$ with $M^w$-clustering $\mathcal{D}^w$ such that $\sum _{y\in\{v, w\}} cost_{G^y}(\mathcal{D}^y) < \sum _{y\in\{v, w\}} cost_{G^y}(\mathcal{C}^y)$. Then, $\{D_i^v\cup D_i^w\}_{i=1}^p$ is a $M^u$-sequencing for $G^u$ since $M^v + M^w = M^u$ and $V^v + V^w = V^u$. Furthermore, the cost of the $M^u$-sequencing $\{D_i^v\cup D_i^w\}_{i=1}^p$ is $\sum _{y\in\{v, w\}} cost_{G^y}(\mathcal{D}^y) + h(M^v, M^w, S)$, where $h(M^v, M^w, S)$ also equals
\begin{align*}
& \sum_{C \in \mathcal{C}^u} |C \cap V^v||C \cap V^w| + \sum_{(i, i')\in S}|V^v_i||V^w_{i'}| - 2\sum_{(i, i')\in S} \sum_{C \in \mathcal{C}^u} |C \cap V^v_i||C \cap V^w_{i'}|
\end{align*}
based on the definitions of $M^v$ and $M^w$. As a result, $\sum _{y\in\{v, w\}} cost_{G^y}(\mathcal{D}^y) + h(M^v, M^w, S) < \sum _{y\in\{v, w\}} cost_{G^y}(\mathcal{C}^y) + h(M^v, M^w, S) = cost_{G^u}(\mathcal{C}^u)$, a contradiction. 
Therefore, $\C^u$ is obtained by merging optimal $M^v$ and $M^w$-clusterings. 
Since our value of $D(u, M^u)$ eventually considers $D(v, M^v) + D(w, M^w)$, which by induction contain the optimal costs, we have that $D(u, M^u)$ is at most the cost of $\C^u$.  It is also easy to see that each possible entry considered in the minimization corresponds to an $M^u$-clustering of $G^u$, which cannot be better than $\C^u$, and so it follows that $D(u, M^u)$ is also at least the cost of $\C^u$.  Thus our value of $D(u, M^u)$ is correct.

Consider an internal node $u$ with corresponding operation $\circ_R$. Let $v$ be the child of $u$. Suppose  $\{C^u_i\}_{i = 1}^p$ is an optimal $M^u$-sequencing for $G^u$ with $M^u$-clustering $\mathcal{C}^u$. We define $M^v = (m^v_{i,j})$ such that $m^v_{i,j} = |V^v_i \cap C^u_j|$. Then, $M^u_{i',:} = \sum_{(i,i')\in R} M^v_{i,:}$ for each $i \in [k]$. 
Now, we claim that $\{C^u_i\}_{i = 1}^p$ is also an optimal $M^v$-sequencing of $G^v$ with $M^v$-clustering $\mathcal{C}^u$. Otherwise, assume for contradiction that $\{C^v_i\}_{i = 1}^p$ is an $M^v$-sequencing of $G^v$ with $M^v$-clustering $\mathcal{C}^v$ such that $cost_{G^v}(\mathcal{C}^v) < cost_{G^v}(\mathcal{C}^u)$. Then, $\{C^v_i\}_{i = 1}^p$ is also a $M^u$-sequencing of $G^u$ with $M^u$-clustering $\mathcal{C}^v$, because (1) the non-empty sets of $\{C^v_i\}_{i = 1}^p$ define a partition of $V^v$, thus also define a partition of $V^u$ since $V^v = V^u$. (2) In $C^v_j$ of $V^v$, the number of vertices labeled $i$ is $m^v_{i,j}$. In addition, $G^u$ is obtained from $G^v$ by relabeling vertices controlled by function $R$. Hence, in $C^v_j$ of $V^u$, the number of vertices labeled $i'$ is $\sum_{(i,i')\in R} m^v_{i,j}$, which equals $m^u_{i',j}$ for each $i'\in [k]$. Thus, the cost of the $M^u$-sequencing $\{C^v_i\}_{i = 1}^p$ is $cost_{G^u}(\mathcal{C}^v) = cost_{G^v}(\mathcal{C}^v) < cost_{G^v}(\mathcal{C}^u) = cost_{G^u}(\mathcal{C}^u)$, a contradiction. As before, this shows that our value of $D(u, M^u)$ is both an upper and lower bound on the cost of $\C^u$, and is therefore correct.
\end{proof}

Recall that cographs have clique-width at most $2$.  Moreover, a $cw(2)$-expression and then a $2$-expression can easily be derived from the cotree (in fact, dynamic programming over the cotree directly could be more efficient).  We therefore have the following consequence, which we note is not conditioned on receiving a $k$-expression.

\begin{corollary}
$p$-\textsc{Cluster Editing} on cographs admits an $O(n^{4p + 4})$ time algorithm.
\end{corollary}

\section{Cluster Editing on Trivially Perfect Graphs}

We now show that \textsc{Cluster Editing} can be solved in cubic time on TPGs, first providing some intuitions.  Consider a TPG $G$, and recall that it admits a cotree $T$ in which every $1$-node has at most one child that is not a leaf.  Such a cotree can be found in linear time \cite{DBLP:journals/siamcomp/CorneilPS85}. Let $\C$ be an optimal clustering of $G$.  Let $v \in V(T)$ and let $X \subseteq V(G)$ be the set of leaves that descend from $v$, which we call a \emph{clade}.  Suppose that there are two clusters $C_1, C_2 \in \C$ that intersect with $X$, but that also satisfy $C_1 \setminus X \neq \emptyset, C_2 \setminus X \neq \emptyset$ (we will say that $X$ ``grows'' in $C_1$ and $C_2$). 
We can show that there is an alternate optimal clustering obtainable by moving $C_1 \cap X$ into a new cluster by itself, and merging $(C_1 \setminus X) \cup C_2$ into a single cluster\footnote{Note that this is where the properties of TPGs are used: this works because if we consider the neighbors of $X$ outside of $X$, they are all children of $1$-nodes on the path from $v$ to the root.  They thus form a clique, which makes the merging $(C_1 \setminus X) \cup C_2$ advantageous as it saves the deletions of edges from that clique.  Also, we may need to exchange the roles of $C_1$ and $C_2$.}.  In this manner, $X$ ``grows'' in one less cluster, and we can repeat this until it grows in at most one cluster.  
This property allows dynamic programming over the clades of the cotree, as we only need to memorize optimal solutions with respect to the size of the only cluster that can grow in the subsequent clades.


For two vertex-disjoint subsets of vertices $X, Y$, we denote $E_G(X, Y) = \{uv \in E(G) : u \in X, v \in Y\}$ and $e_G(X, Y) = |E_G(X, Y)|$.  
We further denote $\overline{e}_G(X, Y) = |X||Y| - e_G(X, Y)$, which is the number of non-edges between $X$ and $Y$.  We drop the subscript $G$ from these notations if it is clear from the context.
Two disjoint subsets $X, Y \subseteq V(G)$ are \emph{neighbors} if $e(X, Y) = |X||Y|$, and they are \emph{non-neighbors} if $\overline{e}(X, Y) = |X||Y|$.  That is, every possible edge is present, and absent, respectively.
Note that using this notation, for a given clustering $\C = \{C_1, \cdots, C_l\}$, the set of inserted edges is $\bigcup_{C \in \C} E(\overline{G[C]})$ and the set of deleted edges is $\bigcup_{1 \leq i < j \leq l} E(C_i, C_j)$.

Let $G$ be a cograph with cotree $T$.  For $v \in V(T)$, we denote by $L(v)$ the set of leaves that descend from $v$ in $T$ (note that $L(v) \subseteq V(G)$).  We call $L(v)$ the \emph{clade of $v$}.  
We say that $X \subseteq V(G)$ is a \emph{clade} of $T$ if $X$ is a clade of some $v \in V(T)$.  In this case, we say that $X$ is \emph{rooted at $v$}.
The set of clades of $T$ is defined as $clades(T) = \{L(v) : v \in V(T)\}$.

We first show a technical property of optimal solutions that will be useful.  The lemma statement is illustrated in Figure~\ref{fig:tpglemma}.

\begin{figure}[ht]
\begin{center}
\includegraphics[width=8cm]{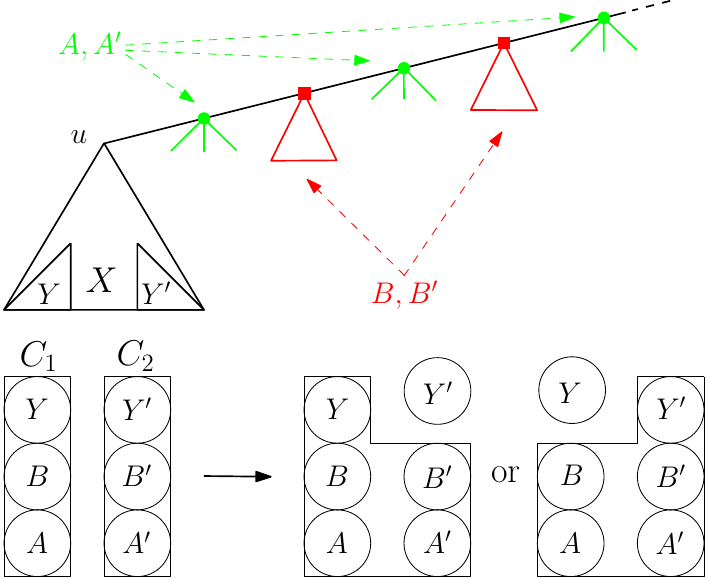}
\caption{Illustration of Lemma~\ref{lem:two-clusters}.  The tree shown is the cotree of $G$, with disc vertices being $1$-nodes and square vertices $0$-nodes. Here, $Y = X \cap C_1$ and $Y' = X \cap C_2$.  If $|A| \geq |B|$ and $|A'| \geq |B'|$, then rearranging $C_1$ and $C_2$ in one of the two ways shown also gives an optimal clustering.}
\label{fig:tpglemma}
    
\end{center}
\end{figure}

\begin{lemma}
\label{lem:two-clusters}
Let $G$ be a TPG with cotree $T$, let $u \in V(T)$, and let $X= L(u)$.  
Let $\mathcal{C}$ be an optimal clustering of $G$ and let $C_1, C_2 \in \mathcal{C}$ be distinct clusters that both intersect with $X$. 
Let $U$ be the set of strict ancestors of $u$ in $T$.  
Let $A$ (resp. $A'$) be the set of vertices of $C_1$ (resp. $C_2$) that are children of $1$-nodes in $U$, and let $B = C_1 \setminus (X \cup A)$ 
(resp. $B' = C_2 \setminus (X \cup A')$).

If $|A| \geq |B|$ and $|A'| \geq |B'|$, then at least one of 
the alternate clusterings $(\mathcal{C} \setminus \{C_1, C_2\}) \cup \{ C_1 \cup A' \cup B', C_2 \cap X \}$ or $(\mathcal{C} \setminus \{C_1, C_2\}) \cup \{ C_1 \cap X, C_2 \cup A \cup B \}$ has cost at most $cost_G(\mathcal{C})$.
\end{lemma}

\begin{proof}
Let us denote $Y = C_1 \cap X$ and $Y' = C_2 \cap X$.  Note that if $A, B, A', B'$ are all empty, the lemma holds, so assume this is not the case.
Observe that for each $1$-node $u_i$ in $U$, the set of leaf children of $u_i$ are part of a critical clique, and must all be in the same cluster of $\mathcal{C}$ by Proposition~\ref{critical-clique-prop}.

Our main challenge is to establish that there are at least as many edges between $A \cup B$ and $A' \cup B'$ than non-edges, after which the statement will follow easily.
The difficulty in proving this claim is that even though $A, A'$ are neighbors, $A, B'$ and $A', B$ could be non-neighbors, possibly creating too many non-edges.  However, there cannot be too many such non-edges.  For instance, the children of the highest $1$-node in $U$ are neighbors of all of $B \cup B'$.  More generally, we can ``layer'' the $A \cup A'$ nodes in order of decreasing neighborhood size in $B \cup B'$, and argue that there are enough edges traversing.

To this end, we set up a series of definitions that are illustrated in Figure~\ref{fig:tpgaibi}.
Let $a_1$ be the lowest common ancestor of $C_1 \cup C_2$ in $T$, which is in $U$ since $C_1 \cup C_2$ intersects with $X$ and some other set $A, A', B, B'$ is non-empty.
If $a_1$ is a $0$-node, the choice of $a_1$ as the lowest common ancestor implies that at least one of $C_1$ or $C_2$ induces a disconnected subgraph, contradicting that $\mathcal{C}$ is an optimal clustering.
Thus $a_1$ is a $1$-node.
Suppose, without loss of generality, that the leaf children of $a_1$ are in $C_1$. 
Then, let $a'_1$ be the $1$-node in $U$ that is the closest descendant of $a_1$, such that the leaf children of $a'_1$ are in $C_2$ (if no such node exists, then $a'_1$ is undefined).  
Then to generalize, assume that $a'_i$ is defined for some $i \geq 1$.  Then let $a_{i+1}$ be the closest $1$-node in $U$ descending from $a'_i$ such that the leaf children of $a_{i+1}$ are all in $C_1$ (or leave $a_{i+1}$ undefined if this does not exist).  Similarly, if $a_i$ is defined, $a'_i$ is the closest descendant of $a_i$ in $U$ such that the leaf children of $a'_i$ are all in $C_2$ (or undefined if non-existent). 

\begin{figure}[ht]
\begin{center}
\includegraphics[width=12cm]{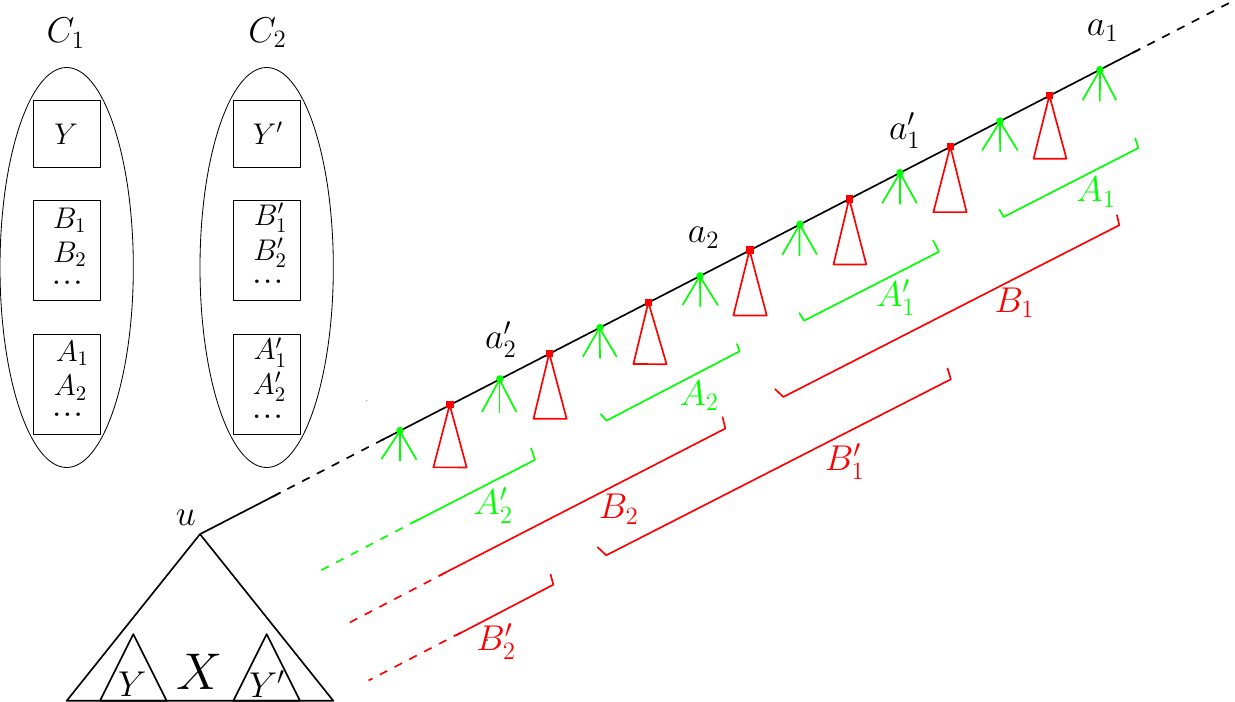}
\caption{Illustration of $a_i, a'_i, A_i, B_i, A'_i, B'_i$. }
\label{fig:tpgaibi}
\end{center}
\end{figure}

Let $r$ be the maximum index such that $a_r$ is defined and $s$ the maximum index such that $a'_s$ is defined.  Note that if $a_i'$ is defined, then $a_i$ is defined, and if $a_{i+1}$ is defined, then $a'_i$ is defined.  Thus $r \in \{s, s+1\}$.  
For $i \in [r]$, we put $next(a_i) = a_{i+1}$ if defined, and $next(a_i) = u$ otherwise.  
Then define 
\begin{align*}
A_i = A \cap (L(a_i) \setminus L(next(a_i))), \quad 
B_i = B \cap (L(a_i) \setminus L(next(a_i)).
\end{align*}
In other words, $A_i$ is the set of $A$ vertices that are leaves whose parent is on the path between $a_i$ and the next $a_{i+1}$ (excluding $a_{i+1}$).  Note that by definition there is no such leaf whose parent is on the path between $a'_i$ and $a_{i+1}$.  Then, $B_i$ are the $B$ vertices whose closest ancestor in $U$ is a $0$-node between $a_i$ and $a_{i+1}$ (note that here, this ancestor could be between $a'_i$ and $a_{i+1}$).

Similarly, for $i \in [s]$, put $next(a'_i) = a'_{i+1}$ if defined, and $next(a'_i) = u$ otherwise.  
Then define 
\begin{align*}
A'_i = A' \cap (L(a'_i) \setminus L(next(a'_i))), \quad 
B'_i = B' \cap (L(a'_i) \setminus L(next(a'_i)).
\end{align*}
For $i, j \in [r]$, let $A[i,j] = A_i \cup A_{i+1} \cup \ldots \cup A_j$ and $B[i,j] = B_i \cup B_{i+1} \cup \ldots B_j$.  
For $i,j \in [s]$, let $A'[i,j] = A'_i \cup \ldots A'_j$ and $B'[i,j] = B'_i \cup \ldots \cup B'_j$.
Note that $B'[1,s] = B'$ since all vertices of $B'$ are descendants of $a'_1$, otherwise, $C_2$ induces a disconnected subgraph, contradicting that $\mathcal{C}$ is an optimal clustering.

Armed with these definitions, we can make the following observations:
\begin{enumerate}
\item 
$A \cup A'$ is a clique.

\item 
$X$ is a neighbor of $A \cup A'$, and a non-neighbor of $B \cup B'$.

\item 
For distinct $i, j$, $B_i, B_j$ are non-neighbors, and $B_i', B_j'$ are non-neighbors.

\item 
$B_i, A[i+1,r]$ are non-neighbors, and $B_i, A'[1,i-1]$ are neighbors.  
Likewise, $B'_i, A'[i+1, s]$ are non-neighbors, and $B'_i, A[1,i]$ are neighbors.
\end{enumerate}

Let us now proceed with an important claim.

\begin{claim}\label{claim:a1ib1i}
For each $i \in [r]$, $|A[1,i]| \geq |B[1,i]|$.  Moreover, for each $i \in [s]$, $|A'[1, i]| \geq |B'[1,i]|$.
\end{claim}

\begin{proof}
We only prove for $i \in [r]$, the other statement can be proved in the same manner.
To ease notation, we denote $a = |A[1,i]|$ and $b = |B[1,i]|$.
Suppose that $a < b$.  Consider splitting $C_1$ into $A[1,i] \cup B[1,i]$ and $A[i+1,r] \cup B[i+1,r] \cup Y$, thereby producing a clustering with $C_1$ replaced by two sets.
The edge modifications paid by this new clustering but not by $\mathcal{C}$ is due to the edges to delete between $A[1,i] \cup B[1, i]$ and $A[i+1, r] \cup B[i+1, r] \cup Y$.
As we just observed, $B[1, i]$ is a non-neighbor of each of $A[i+1, r], B[i+1, r]$ and $X$ (and thus $Y$).  Thus, 
the increase in cost is at most $a (|B| - b) + a(|A| - a) + a|Y|$.  On the other hand, separating $B[1, i]$ from the other subsets saves a cost of at least $b (|B| - b) + b (|A| - a) + b |Y|$, which were insertions in $\C$.  Assuming $a < b$ means that the increase in cost is strictly smaller
than the decrease in cost, contradicting that $\mathcal{C}$ is an optimal clustering.
\end{proof}

We can next show that there are at least as many edges traversing $A \cup B$ and $A' \cup B'$ than non-edges.

\begin{claim}\label{claim:edge-vs-nonedge}
$e(A \cup B, A' \cup B') \geq \overline{e}(A \cup B, A' \cup B').$
\end{claim}

\begin{proof}
We first argue that
\begin{align*}
e(A \cup B, A' \cup B') = &~\sum_{i=1}^s e(A'_i \cup B'_i, A[1, i] \cup B[1,i])~+ \\
				    &~ \sum_{i=2}^r e(A_i \cup B_i, A'[1,i-1] \cup B'[1,i-1]).
\end{align*}
Indeed, consider an edge $pq$ of $E(A \cup B, A' \cup B')$ such that $p \in A_i \cup B_i$ and $q \in A'_j \cup B'_j$ for some $i, j$.  
If $i > j$, then $pq$ is counted once in the second summation when adding $e(A_i \cup B_i, A'[1,i-1] \cup B'[1,i-1])$.  Moreover, it is not counted in the first summation.  If $j \geq i$, then $pq$ is counted once in the first summation and not counted in the second.  
Using the same arguments, we see that 
\begin{align*}
\overline{e}(A \cup B, A' \cup B') = &~\sum_{i=1}^s \overline{e}(A'_i \cup B'_i, A[1, i] \cup B[1,i]) + \\
				    &~ \sum_{i=2}^r \overline{e}(A_i \cup B_i, A'[1,i-1] \cup B'[1,i-1]).
\end{align*}
We next argue that $e(A'_i \cup B'_i, A[1, i] \cup B[1,i]) \geq \overline{e}(A'_i \cup B'_i, A[1,i] \cup B[1,i])$ for each $i \in [s]$.
To see this, recall the $A'_i, A[1,i]$ are neighbors and $B'_i, A[1,i]$ are neighbors.  This means that the number of edges between $A'_i \cup B'_i$ and $A[1, i] \cup B[1, i]$ is at least $|A[1,i]| (|A'_i| + |B'_i|)$, and the number of non-edges is therefore at most $|B[1,i]| (|A'_i| + |B'_i|)$.  By Claim~\ref{claim:a1ib1i}, $A[1,i] \geq B[1,i]$, implying that there are at least as many edges as non-edges.
Using the same arguments, we can obtain $e(A_i \cup B_i, A'[1, i-1] \cup B'[1,i-1]) \geq \overline{e}(A_i \cup B_i, A'[1, i-1] \cup B'[1,i-1])$ 
(this follows from the fact that $A'[1,i-1]$ is a neighbor of both $A_i$ and $B_i$ and again Claim~\ref{claim:a1ib1i}).  

We conclude that each term in the $e(A \cup B, A' \cup B')$ summation is at least as large as the corresponding term in the $\overline{e}(A \cup B, A' \cup B')$, which proves our claim.
\end{proof}

Claim~\ref{claim:edge-vs-nonedge} allows us to finally prove the lemma.
First suppose that $|Y| \geq |Y'|$.
Consider the clustering obtained from $\mathcal{C}$ by replacing $C_1, C_2$ by two other subsets $Y \cup A \cup B \cup A' \cup B'$ and $Y'$.
Compared to $\mathcal{C}$, this clustering has an increase in cost of at most $\overline{e}(A \cup B, A' \cup B') + |Y||B'| + |Y'||A'|$ (recalling that $Y, A'$ are neighbors, and $Y', B'$ are non-neighbors).
On the other hand, there is a decrease in cost of $e(A \cup B, A' \cup B') + |Y||A'| + |Y'||B'|$.
The decrease versus increase difference in cost is
\begin{align*}
	&~e(A \cup B, A' \cup B') + |Y||A'| + |Y'||B'| - \overline{e}(A \cup B, A' \cup B') - |Y||B'| - |Y'||A'| \\
       =&~e(A \cup B, A' \cup B') - \overline{e}(A \cup B, A' \cup B') + (|A'| - |B'|)(|Y| - |Y'|).
\end{align*}
Since $e(A \cup B, A' \cup B') - \overline{e}(A \cup B, A' \cup B') \geq 0$ by Claim~\ref{claim:edge-vs-nonedge}, $|A'| \geq |B'|$ is assumed in the lemma statement, and $|Y| \geq |Y'|$ was just assumed, this difference is non-negative.  It follows that the alternate clustering has an equal or better cost.  
The case $|Y'| \geq |Y|$ can be handled in the exact same manner by replacing $C_1, C_2$ by $Y, Y' \cup A' \cup B' \cup A \cup B$.
\end{proof}

Let $G$ be a TPG with cotree $T$.
Let $\mathcal{C} = \{C_1, \ldots, C_l\}$ be a clustering of $G$.
Let $X$ be a clade of $T$.  For $C_i \in \mathcal{C}$, we say that $X$ \emph{grows in $C_i$} if $C_i \cap X \neq \emptyset$ and $C_i \setminus X \neq \emptyset$.  Note that this differs from the notion of overlapping, since $X \subset C_i$ is possible. 
We then say that $X$ has a \emph{single-growth in $\mathcal{C}$} if $X$ grows in at most one element of $\mathcal{C}$.  In other words, at most one cluster of $\mathcal{C}$ containing elements of $X$ also has elements outside of $X$, and the rest of $X$ is split into clusters that are subsets of $X$.
Furthermore, we will say that $X$ has a \emph{heritable single-growth in $\mathcal{C}$} if, for all clades $Y$ of $T$ such that $Y \subseteq X$, $Y$ has a single-growth in $\mathcal{C}$.  For $v \in V(T)$, we may also say that $v$ grows in $C_i$ if $L(v)$ grows in $C_i$, or that $v$ has a single-growth in $\mathcal{C}$ if $L(v)$ does.
For brevity, we may write SG for single-growth, and HSG for heritable single-growth.

\begin{lemma}
    Suppose that $G$ is a TPG with cotree $T$.
    Then there exists an optimal clustering $\mathcal{C}$ of $G$ such that every clade $X$ of $T$ has an SG in $\mathcal{C}$.
\end{lemma}

\begin{proof}
    Consider an optimal clustering $\mathcal{C} = \{C_1, \ldots, C_l\}$ of $G$, chosen such that the number of clades of $T$ that have an HSG in $\mathcal{C}$ is maximum, among all optimal clusterings\footnote{We could attempt to choose $\mathcal{C}$ to maximize the clades with an SG instead, but we will modify $\mathcal{C}$ later on, and keeping track of changes in HSG clades is much easier than SGs.}.  
    If every clade of $T$ has the HSG property, then we are done (since HSG implies SG), so we assume that not every clade has the HSG property. Clearly, in $\mathcal{C}$, every leaf of $T$ has an HSG, the root of $T$ does not have an HSG, and there exists at least one internal node of $T$ that does not have an SG. Choose $u \in V(T)$ with the minimum number of descendants such that $u$ does not have an SG in $\mathcal{C}$. Then, every descendant of $u$ has an SG, thus, also has an HSG in $\mathcal{C}$.
    Notice that $u$ cannot be the root of $T$, because the root trivially has an SG in $\mathcal{C}$.

    Denote $X = L(u)$.  Since $u$ does not have the SG property, $X$ grows in at least two clusters of $\mathcal{C}$, say $C_1$ and $C_2$.
    Thus $X \cap C_1, C_1 \setminus X, X \cap C_2, C_2 \setminus X$ are all non-empty.
    Denote $Y = X \cap C_1, Y' = X \cap C_2$.  
    We show that we can transform $\mathcal{C}$ into another optimal clustering $\mathcal{C}'$ in which one of $Y$ or $Y'$ is a cluster by itself, and such that the number of clades that have an HSG in $\mathcal{C}'$ is no less than in $\mathcal{C}$.

    As in Lemma~\ref{lem:two-clusters}, let $U$ be the set of strict ancestors of $u$ in $T$. Let $A$ and $A'$ be the sets of vertices of $C_1$ and $C_2$, respectively, that are children of $1$-nodes in $U$.  Then let $B = C_1 \setminus (X \cup A)$ and $B' = C_2 \setminus (X \cup A')$.  Note that because $C_1 \setminus X \neq \emptyset$, $A \cup B$ is non-empty.  Likewise, $A' \cup B'$ is non-empty.

   We argue that $|A| \geq |B|$.  Suppose instead that $|A| < |B|$.
    Consider the clustering  $\mathcal{C}' = (\mathcal{C} \setminus \{C_1\}) \cup \{Y, A \cup B\}$.  
    Then $\mathcal{C}'$ has $|A||Y|$ edge deletions that are not in $\mathcal{C}$ (and no other edge modification is in $\mathcal{C}'$ but not in $\mathcal{C}$, since $Y$ and $B$ share no edge).  On the other hand, $\mathcal{C}$ has $|B||Y| > |A||Y|$ edge additions that are not in $\mathcal{C}'$.  Hence, the cost of $\mathcal{C}'$ is strictly lower than that of $\mathcal{C}$, a contradiction.
   By the same argument, we get that $|A'| \geq |B'|$.

   We see that $C_1$ and $C_2$ satisfy all the requirements of Lemma~\ref{lem:two-clusters}, and so we may get an alternate optimal clustering $\mathcal{C}'$ or $\mathcal{C}''$, where 
$\mathcal{C}'$ is obtained from $\mathcal{C}$ by replacing $C_1, C_2$ by $Y, C_2 \cup A \cup B$, and $\mathcal{C}''$ is obtained by replacing $C_1, C_2$ by $Y', C_1 \cup A' \cup B'$.  Notice that $X$ grows in fewer clusters in $\mathcal{C}'$ than in $\mathcal{C}$, and the same is true for $\mathcal{C}''$.  Before proceeding, we need to argue that $T$ has as many clades with an HSG in $\mathcal{C}'$, and in $\mathcal{C}''$ than in $\mathcal{C}$.

So, let $w \in V(T)$ be such that $w$ has an HSG in $\mathcal{C}$.  
Observe that $w$ cannot be in $U \cup \{u\}$, since these have $u$ as a descendant and $u$ does not have an SG in $\mathcal{C}$.  
Therefore, $w$ must either be: (1) a leaf child of a $1$-node in $U$; (2) a descendant of $u$; or (3) a node whose first ancestor in $U$ is a $0$-node.
Let $v$ be a descendant of $w$, with $v = w$ possible.  By the definition of HSG, $v$ has an SG in $\mathcal{C}$.  In all cases, we argue that $v$ still has an SG in $\mathcal{C}'$ and $\mathcal{C}''$.
\begin{enumerate}
    \item 
    If $w$ is the child of a $1$-node of $U$, then $w = v$ is a leaf and it trivially has an SG in $\mathcal{C}'$ and $\mathcal{C}''$.

    \item 
    Suppose that $w$, and thus $v$, is a descendant of $u$.  
    If $L(v)$ does not intersect with $C_1$ nor $C_2$, then the clusters of $\mathcal{C}$ that intersect with $L(v)$ are unaltered in $\mathcal{C}'$ and $\mathcal{C}''$, and thus $v$ also has an SG in $\mathcal{C}'$, and in $\mathcal{C}''$.  Thus suppose that $L(v)$ intersects with $C_1 \cup C_2$.
    
    In that case, since $v$ descends from $u$, $L(v) \subseteq X$ and it must thus intersect with $Y \cup Y'$.  
    If $L(v) \cap Y \neq \emptyset$, then $v$ grows in $C_1$ because $A \cup B \neq \emptyset$.  Likewise, if $L(v) \cap Y' \neq \emptyset$, then $v$ grows in $C_2$.  It follows that $L(v)$ intersects exactly one of $Y$ or $Y'$, and grows in exactly one of $C_1$ or $C_2$.  
    If $L(v)$ intersects $Y$, then in $\mathcal{C}'$, $L(v)$ may grow in the cluster $Y$, but it does not grow in $C_2 \cup A \cup B$ since it does intersect with it, and does not grow in other clusters of $\mathcal{C}'$ since these were unaltered.  Thus $L(v)$ has an SG in $\mathcal{C}'$.  In $\mathcal{C}''$, $L(v)$ grows in $C_1 \cup A' \cup B'$ but no other cluster for the same reason.  Thus $L(v)$ has an SG in $\mathcal{C}''$ as well.  
If $L(v)$ intersects $Y'$, the same arguments can be used to deduce that $L(v)$ has an SG in $\mathcal{C}'$ and $\mathcal{C}''$.

    \item 
    Finally, suppose that $w$ is such that its first ancestor in $U$ is a $0$-node.  Again we may assume that $L(v)$ intersects $C_1 \cup C_2$.  
    This implies that $L(v)$ intersects with $B \cup B'$, and does not intersect with $Y \cup Y' \cup A \cup A'$.  If $L(v)$ intersects $B$, then it grows in $C_1$ since $|A| \geq |B|$ and thus $A$.  Then in $\mathcal{C}'$, $v$ grows in $C_2 \cup A \cup B$, but not $Y$ nor any other unaltered cluster.
    In $\mathcal{C}''$ it grows in $C_1 \cup A' \cup B'$, but not $Y'$ nor any other unaltered cluster.  If $L(v)$ intersects $B'$, the same argument applies.  Either way, $L(v)$ has an SG in $\mathcal{C}'$ and $\mathcal{C}''$.

\end{enumerate}

  Since any descendant of $w$ has an SG in either $\mathcal{C'}$ and $\mathcal{C}''$, we deduce that $w$ has an HSG in both alternate clusterings.
 
We have thus found an optimal clustering $\mathcal{C}^* \in \{\mathcal{C}', \mathcal{C}''\}$ such that every $w \in V(T)$ that has an HSG in $\mathcal{C}$ also has an HSG in $\mathcal{C}
^*$.
Moreover, since $\mathcal{C}^*$ ``extracts'' either $Y$ or $Y'$ from its cluster, $X$ grows in one less cluster of $\mathcal{C}^*$.  If $X$ grows in only one such cluster, then $u$ has an SG in $\mathcal{C}^*$ and therefore also an HSG in $\mathcal{C}^*$, by the choice of $u$.  
In this case, $T$ has more clades that have an HSG in $\mathcal{C}^*$ than with $\mathcal{C}$, which contradicts our choice of $\mathcal{C}$.  If $X$ still grows in at least two clusters of $\mathcal{C}^*$, we may repeat the above modification as many times as needed until $X$ grows in a single cluster, yielding the same contradiction.  We deduce that every node of $T$ has an HSG, and therefore an SG in $\mathcal{C}$.
\end{proof}

Our goal is to use the above to perform dynamic programming over the cotree.  For a node $v$ of this cotree, we will store the value of a solution for the subgraph induced by $L(v)$, and will need to determine which cluster of such a partial solution should grow.  As it turns out, we should choose the largest cluster to grow.

\begin{lemma}
\label{SG-cluster-is-the-largest}
Let $G$ be a TPG with cotree $T$.  Let $\mathcal{C}$ be an optimal clustering of $G$ such that every clade of $T$ has an SG in $\mathcal{C}$ and, among all such possible optimal clusterings, such that $|\mathcal{C}|$ is maximum.  
Then for every clade $X$ of $T$, one of the following holds:
\begin{itemize}
\item 
$X$ does not grow in any $C_i \in \mathcal{C}$; or

\item 
$X$ grows in one $C_i \in \mathcal{C}$, and $|X \cap C_i| = \max_{C_j \in \mathcal{C}} |X \cap C_j|$.
\end{itemize}
\end{lemma}

\begin{proof}
Let $X$ be a clade of $T$ and suppose that $X$ grows in some $C_i \in \mathcal{C}$.  Note that vertices of $X$ share the same neighborhood outside of $X$.  
Thus $C_i$ can be partitioned into $\{X \cap C_i, A, B\}$ such that 
$X, A$ are neighbors and $X, B$ are non-neighbors.  Note that $|A| \geq |B|$ as otherwise, we could obtain an alternate clustering $\mathcal{C}'$ by replacing $C_i$ by $\{X \cap C_i, A \cup B\}$ and save a cost of $|X \cap C_i|(|B| - |A|) > 0$.

We also argue that $|A| > |B|$.  This is because if $|A| = |B|$, the same clustering $\mathcal{C}'$ has $cost_G(\mathcal{C}') = cost_G(\mathcal{C})$, but has one more cluster. 
We also argue that every clade of $T$ has an SG in $\mathcal{C}'$, contradicting our choice of $\mathcal{C}$. Let $Y = X \cap C_i$  and consider some $v\in V(T)$.  By assumption $v$ has an SG in $\mathcal{C}$.  If $C_i \cap L(v) = \emptyset$, then the clusters of $\mathcal{C}$ that intersect with $L(v)$ are unaltered in $\mathcal{C}'$ and $v$ also has an SG $\mathcal{C}'$.  Suppose that $C_i \subseteq L(v)$.  Apart from $C_i = Y\cup A\cup B$, the clusters of $\mathcal{C}$ that intersect with $L(v)$ are unaltered in $\mathcal{C}'$. Moreover, $L(v)$ does not grow in $Y \cup A \cup B$, and neither does it grow in $Y$ or $A\cup B$.  Thus $v$ also has an SG $\mathcal{C}'$.  The remaining case is when $L(v)$ grows in $C_i$ (but no other cluster of $\mathcal{C}$).  Let $u \in V(T)$ be such that $L(u) = X$.  If $v = u$, then $L(v)$ does not grow in any cluster of $\mathcal{C}'$. If $v$ is a strict descendant of $u$, then $L(v)$ can only grow in $Y$ of $\mathcal{C}'$. If $v$ is a strict ancestor of $u$, then $Y \subseteq L(v)$ and $L(v)$ can only grow in $A\cup B$ of $\mathcal{C}'$. If $v$ is in the rest of $V(T)$, then $Y \cap L(v) = \emptyset$ and $L(v)$ grows only in $A\cup B$ of $\mathcal{C}'$. As a result, $L(v)$ has an SG in $\mathcal{C}'$.
Since this holds for any $v$, every node has an SG in $\mathcal{C'}$, which is a contradiction since $|\mathcal{C}| < |\mathcal{C}'|$.  Therefore, $|A| > |B|$.

Now suppose that there is some $C_j \neq C_i$ such that $|X \cap C_j| > |X \cap C_i|$.  Because $X$ already grows in $C_i$, the SG property implies that $C_j \subseteq X$.  Consider the alternate clustering $\C^*$ obtained from $\C$ by replacing $C_i, C_j$ by $C_j \cup A \cup B, X \cap C_i$.  The number of modifications in $\C^*$ but not in $\C$ is $|X \cap C_i||A| + |C_j||B|$, but the number of modifications in $\C$ not in $\C^*$ is $|X \cap C_i||B| + |C_j||A|$.  The difference between the latter and the former is $(|C_j| - |X \cap C_i|)(|A| - |B|)$.  Since $|A| > |B|$ and $|C_j| = |X \cap C_j| > |X \cap C_i|$, this is greater than $0$, and thus $\C^*$ is a clustering of cost lower than $\mathcal{C}$, a contradiction.
\end{proof}

Our algorithm will search for an optimal clustering that satisfies all the requirements of Lemma~\ref{SG-cluster-is-the-largest}.  That is, for a TPG $G$ with cotree $T$, a clustering $\mathcal{C}$ is \emph{well-behaved} if it is optimal, every clade of $T$ has an SG in $\mathcal{C}$, and among all such possible clusterings $|\mathcal{C}|$ is maximum.  As we know that such a $\mathcal{C}$ exists, we will search for one using dynamic programming.

In the remainder, for an arbitrary clustering $\mathcal{C}$ of $G$ and $X \subseteq V(G)$, define $\mathcal{C}|_X = \{ C_i \cap X | C_i \in \mathcal{C}\} \setminus \{ \emptyset \}$, i.e., the restriction of $\mathcal{C}$ to $X$.
Note that $\mathcal{C}|_X$ is a clustering of $G[X]$, and we refer to $cost_{G[X]}(\mathcal{C}|_X)$ as the \emph{cost of $\mathcal{C}$ in $G[X]$}.
Although $\mathcal{C}|_X$ is not necessarily an optimal clustering of $G[X]$, we can deduce from the above that it has minimum cost among those clusterings with the same largest cluster.

\begin{corollary}
\label{cor:largest-optimal}
    Let $G$ be a TPG with cotree $T$, and let $\mathcal{C}$ be a well-behaved clustering of $G$.  Let $u \in V(T)$ with clade $X = L(u)$, and let $k_u^* = \max_{C_j \in \mathcal{C}} |X \cap C_j|$ denote the size of the largest intersection of $\mathcal{C}$ with $X$. 
    
    Then $\mathcal{C}|_X$ is a clustering of $G[X]$ such that $cost_{G[X]}(\mathcal{C}|_X)$ is minimum, among all clusterings of $G[X]$ whose largest cluster has size $k_u^*$.
\end{corollary}

\begin{proof}
Let $\C'$ be a clustering of $G[X]$ whose largest cluster has size $k^*_u$, and assume that $cost_{G[X]}(\C') < cost_{G[X]}(\C|_X)$. 
Suppose first that $X$ does not grow in any cluster of $\C$.  Then there is $\{C_1, \ldots, C_l\} \subseteq \C$ that forms a partition of $X$.  Because all the elements of $X$ have the same neighbors outside of $X$, we can replace $\{C_1, \ldots, C_l\}$ by $\C'$ in $\C$, which incurs fewer edit operations between elements of $X$, but does not alter edit operations between other pairs of vertices.  This contradicts the optimality of $\C$.  

Otherwise by Lemma~\ref{SG-cluster-is-the-largest}, $X$ grows in one cluster $C_i \in \C$ that maximizes $|C_i \cap X|$, which is equal to $k^*_u$.  Moreover, there are clusters $C_1, \ldots, C_l$ such that 
$\{C_1, \ldots, C_l\} \cup \{C_i \cap X\}$ forms a partition of $X$.  
Let $C_i'$ be a cluster of $\C'$ of size $k^*_u$.  In $\C$, we can replace $\{C_1, \ldots, C_l\}$ by $\C' \setminus \{C_i'\}$, and replace $C_i$ by $(C_i \setminus X) \cup C_i'$.  Again because vertices of $X$ have the same neighbors outside of $X$, this incurs less cost between pairs in $X$, and the same cost between other pairs of vertices, contradicting the optimality of $\C$.
\end{proof}

We define a 2-dimensional dynamic programming table $D$ indexed by $V(T) \times [n]$, with the intent that $D[u, k]$ has the cost of an optimal clustering of $G[L(u)]$ in which the largest cluster has size $k$.  Notice that this intent is mostly for intuition purposes, since we will not be able to guarantee that $D[u, k]$ stores the correct value for each $u, k$.  However, we will 
argue that those entries that correspond to an optimal clustering are correct.  

We assume that the cotree $T$ of $G$ is binary (note that such a cotree always exists and, since the previous lemmas make no assumption on the structure of the cotree, this is without loss of generality).
If $u$ is a leaf, put $D[u, 1] = 0$ and $D[u, k] = \infty$ for every $k \neq 1$.
For internal node $u$, let $u_1, u_2$ be the two children of $u$ in $T$.  We put
\begin{align*}
D[u, k] = \min \begin{cases}
	D[u_1, k] + \min_{j \in [k]} D[u_2, j] + \mathbb{I}_u \cdot |L(u_1)||L(u_2)| \\
	D[u_2, k] + \min_{j \in [k]} D[u_1, j] + \mathbb{I}_u \cdot |L(u_1)||L(u_2)|\\
	\min_{j \in [k-1]} (D[u_1, j] + D[u_2, k - j] + \alpha(u)),
\end{cases}
\end{align*}
where $\mathbb{I}_u = 0$ if $u$ is a $0$-node and $\mathbb{I}_u = 1$ if it is a $1$-node, 
and
\begin{align*}
\alpha(u) = \min \begin{cases}
	|L(u_1)||L(u_2)| - j(k - j) & \mbox{if $u$ is a $1$-node} \\
	j(k - j)  & \mbox{otherwise.}
\end{cases}
\end{align*}

The recurrence for $D[u, k]$ mainly says that there are three ways to obtain a solution with a largest cluster of size $k$: either that cluster is taken directly from the solution at $u_1$, from the solution at $u_2$, or we take a cluster of size $j$ from $u_1$ and size $k - j$ from $u_2$, and merge them together.

\begin{lemma}
\label{cotree-and-Dtable}
Let $\mathcal{C}$ be a well-behaved clustering of $G$.
Let $u \in V(T)$, and denote by $k^*_u$ the size of the largest cluster of $\mathcal{C}|_{L(u)}$.  Then both of the following hold:
\begin{itemize}
\item 
for each $k$ such that $D[u, k] \neq \infty$, 
there exists a clustering of $G[L(u)]$ of cost at most $D[u, k]$ whose largest cluster has size $k$;

\item 
$D[u, k^*_u]$ is equal to $cost_{G[L(u)]}(\C|_{L(u)})$, the cost of $\mathcal{C}$ restricted to $G[L(u)]$.
\end{itemize}
\end{lemma}

\begin{proof}
We prove the statement by induction on the cotree $T$.  For a leaf $u$, it is clear that both statements hold with $D[u, 1] = 0$.
Let $u$ be an internal node of $T$ and let $u_1, u_2$ be its children.  Denote $X = L(u), X_1 = L(u_1), X_2 = L(u_2)$.

We focus on the first part and assume that $D[u, k] \neq \infty$.  The value of $D[u, k]$ is the minimum among three cases. If $D[u, k] = D[u_1, k] + \min_{j \in [k]} D[u_2, j] + \mathbb{I}_u |L(u_1)||L(u_2)|$, then because $D[u, k] \neq \infty$, $D[u_1, k] \neq \infty$ and $D[u_2, j] \neq \infty$ for the chosen $j$ in the minimum expression.
We can take the disjoint union of a clustering of $G[L(u_1)]$ whose largest cluster has size $k$ (one exists of cost at most $D[u_1, k]$ by induction), and a clustering of $G[L(u_2)]$ with largest cluster of size $j \leq k$ (one exists of cost at most $D[u_2, j]$ by induction).  
If $u$ is a $1$-node, all edges between $L(u_1)$ and $L(u_2)$ are present and $\mathbb{I}_u |X_1||X_2|$ must be added to the cost (whereas if $u$ is a $0$-node, no additional cost is required).  This confirms that, if the first case of the recurrence is the minimum, there exists a clustering with cost at most $D[u, k]$ whose largest cluster has size $k$.  
The same argument applies to the second case of the recurrence.  

So assume that $D[u, k] = \min_{j \in [k-1]} (D[u_1, j] + D[u_2, k - j] + \alpha(u))$.
This case corresponds to taking, by induction, a clustering of $G[X_1]$ and of $G[X_2]$ with the largest cluster of size $j$ and $k - j$, respectively, and merging their largest cluster into one cluster of size $k$ (leaving the other clusters intact).  If $u$ is a $1$-node, the number of deleted edges  between the resulting clusters is $|L(u_1)||L(u_2)| - j(k - j)$, and if $u$ is a $0$-node we must pay $j(k-j)$ edge insertions.  This shows the first part of the statement. 

For the second part, suppose that $k = k^*_u$.  Here, $\mathcal{C}$ is the optimal clustering stated in the lemma. 
As we just argued, there is a clustering of $G[X]$ of cost at most $D[u, k^*_u]$ with the largest cluster size $k^*_u$.  By Corollary~\ref{cor:largest-optimal}, $\C|_X$ has minimum cost among such clusterings, and so $D[u, k^*_u]$ is at least $cost_{G[X]}(\C|_X)$.  We argue that the latter is also an upper bound on $D[u, k^*_u]$.
Let $C_1 \in \mathcal{C}|_{X_1}$, and note that if $C_1 \notin \mathcal{C}|_X$, then $C_1$ was ``merged'' with some cluster of $\mathcal{C}|_{X_2}$ to obtain $\mathcal{C}|_X$.
In this case, $u_1$ grows in some cluster of $\mathcal{C}|_X$, and therefore also grows in some cluster of $\mathcal{C}$.  
By the SG property, this means that there is at most one $C_1$ of $\mathcal{C}|_{X_1}$ such that $C_1 \notin \mathcal{C}|_X$.  
For the same reason, there is at most one $C_2$ of $\mathcal{C}|_{X_2}$ that is not in $\mathcal{C}|_X$.
This in turn implies that either $\mathcal{C}|_X = \mathcal{C}|_{X_1} \cup~ 
 \mathcal{C}|_{X_2}$, or that there is a unique $C_1 \in \mathcal{C}|_{X_1}$ and a unique $C_2 \in \mathcal{C}|_{X_2}$ such that $C_1 \cup C_2$ is in $\mathcal{C}|_X$.
 We now consider both cases.

 \begin{itemize}
     \item 
     If $\C|_X = \C|_{X_1} \cup \C|_{X_2}$, since $\mathcal{C}|_X$ has its largest cluster of size $k = k_u^*$, one of $\mathcal{C}|_{X_1}$ or $\mathcal{C}|_{X_2}$ must have a largest cluster of size $k$, and the other a largest cluster of some size $j \in [k]$.  
     Using induction on the second part of the lemma, the corresponding entries $D[u_i, k]$ and $D[u_{i'}, j]$ for $\{i, i'\} = \{1, 2\}$ store the costs of $\C|_{X_1}$ and $\C|_{X_2}$, and since all the possibilities are considered by the $D[u, k^*_u]$ recurrence, it is clear that $D[u, k^*_u]$ is at most $cost_{G[X]}(\C|_X)$. 

     \item 
     If $C_1 \cup C_2 \in \C|_X$, we have $\mathcal{C}|_X = ((\mathcal{C}|_{X_1} \cup \mathcal{C}|_{X_2}) \setminus \{C_1, C_2\}) \cup \{C_1 \cup C_2\}$.
     Observe that $X_1$ grows in $C_1 \cup C_2$.  This means that $X_1$ also grows in the cluster of $C' \in \mathcal{C}$ that contains $C_1 \cup C_2$.  By Lemma \ref{SG-cluster-is-the-largest}, $|X_1 \cap C'| \geq |X_1 \cap C''|$ for each $C'' \in \mathcal{C}$. 
Since $C'$ contains $C_1$, we get that $|X_1 \cap C'| = |C_1|$, and since $\{ X \cap C''  : C'' \in \mathcal{C}\}$ contains all the clusters of $\mathcal{C}|_{X_1}$, it must be that $C_1$ is the largest cluster of $\mathcal{C}|_{X_1}$.  By the same argument, $C_2$ is the largest cluster of $\mathcal{C}|_{X_2}$.
Let $j = |C_1|$.  We have that $\mathcal{C}|_X$ has the largest cluster size $k$, and is obtained by taking a clustering of $G[L(u_1)]$ with the largest cluster of size $j$, 
and a clustering of $G[L(u_2)]$ with largest cluster of size $k - j$, and merging these two clusters.  If $u$ is a $1$-node, the cost of this is
the cost of $\mathcal{C}$ in $G[L(u_1)]$, which is $D[u_1, j]$ by induction, plus the cost of $\mathcal{C}$ in $G[L(u_2)]$, which is $D[u_2, k - j]$ by induction, 
plus the cost for all the edges between $L(u_1)$ and $L(u_2)$, excluding those between $C_1$ and $C_2$.  If $u$ is a $0$-node, the cost is the same, except that we pay $j (k - j)$ for the non-edges between $C_1$ and $C_2$.  Either way, this case is considered by our recurrence, and we get $D[u, k] \leq cost_{G[X]}(\mathcal{C}|_X)$.
 \end{itemize}
 Having shown both the lower and upper bounds, we get that $D[u, k_u^*] = cost_{G[X]}(\mathcal{C}|_X)$.
\end{proof}

\begin{theorem}
\label{thmtpg}
The \textsc{Cluster Editing} problem can be solved in time $O(n^3)$ on trivially perfect graphs.
\end{theorem}

\begin{proof}
Note that the size of a cluster is in $[1, n]$, and so in the $D[u, k]$ table, we only consider entries with $k \in [n]$.
We can compute each $D[u, k]$ in time $O(n)$, and there are $O(n^2)$ such entries.  
Once all entries are computed, letting $r$ be the root of $T$, we return $\min_{k \in [n]} D[r, k]$.  This takes total time $O(n^3)$.

To see that this algorithm is correct, let $\mathcal{C}$ be a well-behaved clustering of $G$. Let $k^*$ be the size of the largest cluster of $\mathcal{C}$.
By Lemma~\ref{cotree-and-Dtable}, $D[r, k^*] = cost_G(\mathcal{C})$.  Hence, our returned value is at most $cost_G(\mathcal{C})$.  By the same lemma, for any $k \in [n]$ such that $D[r, k] \neq \infty$, there exists a clustering $\mathcal{C}_k$ of $G$ of cost at most $D[r, k]$ whose largest cluster size is $k$.  
By the optimality of $\mathcal{C}$, we have $D[r, k] \geq cost_G(\mathcal{C}_k) \geq cost_G(\mathcal{C}_{k^*}) = cost_G(\mathcal{C})$ for every $k$.  Hence our returned value is at least $cost_G(\mathcal{C})$.  Therefore, we return $cost_G(\mathcal{C})$.
\end{proof}

\vspace{2mm}

\noindent
\textbf{Open problems.}
We observe that the structural properties shown on TPGs only work if the number of desired clusters is unrestricted.  The complexity of  $p$-\textsc{Cluster Editing} on TPGs is open (note that if $p$ is constant, our $n^{O(p)}$ time algorithm on cographs provides a polynomial-time algorithm). Regarding our clique-width (or rather NLC-width), it might be possible to improve the complexity, for example by achieving $n^{O(p + cw)}$ instead of $n^{O(p \cdot cw)}$.

We proved the problem in P on $\{P_4, C_4\}$-free graphs, but we do not know the complexity on $\{P_4, 2K_2\}$-graphs.  
The complement of such graphs are $\{P_4, C_4\}$-free and may be in P as well, but it is unclear whether the editing problem on the complement can be solved using our techniques.  
More generally, it would be ideal to aim for a dichotomy theorem for forbidden induced subgraphs, that is, to characterize the forbidden induced subgraphs that make \textsc{Cluster Editing} in P, and the ones that make it NP-hard.



\bibliographystyle{abbrv}
\bibliography{references}

\newpage 

\appendix

\end{document}